\title{Intersection graphs of segments and $\exists\R$}
\author{
{\sc Ji\v{r}\'{\i} Matou\v{s}ek}\thanks{Supported
by the  ERC Advanced Grant No.~267165
and by the project CE-ITI (GACR P202/12/G061).}
\\
   {\footnotesize Department of Applied Mathematics}\\[-1.5mm]
   {\footnotesize  Charles University, Malostransk\'{e} n\'{a}m. 25}\\[-1.5mm]
{\footnotesize  118~00~~Praha~1,
   Czech Republic, and}\\
{\footnotesize    Institute of  Theoretical Computer Science}\\[-1.5mm]
{\footnotesize    ETH Zurich,
      8092 Zurich, Switzerland}
}
\date{}
\newif\ifafour
\afourtrue

\documentclass[11pt]{article}

\usepackage{a4}

\usepackage{graphicx,url}
\usepackage{color}

\usepackage{amsmath,amsfonts,amssymb,dsfont,theorem}

\newtheorem{theorem}{Theorem}[section]

\newtheorem{lemma}[theorem]{Lemma}

\newtheorem{proposition}[theorem]{Proposition}

\newtheorem{exercise}[theorem]{Exercise}
\newtheorem{problemo}[theorem]{Problem}

\newcommand{\qed}{\hfill\ensuremath{\Box}}
\newenvironment{proof}{\noindent\textbf{Proof.}
}{\qed\par\medskip}

\newenvironment{proofof}[1]{\medskip\noindent\textbf{Proof of #1.}
}{\qed\par\medskip}
\newenvironment{proofhd}[1]{\noindent\textbf{#1.}}{\qed\par\medskip}
{\ifx&#1&%
  \begin{problemo}\else\begin{problemo}[#1]\fi\upshape}%
    {\end{problemo}}

\newcommand{\Z}{\ensuremath{\mathds Z}}

\newcommand{\R}{\ensuremath{\mathds R}}

\newcommand{\PP}{\ensuremath{\mathcal P}}
\newcommand{\DD}{\ensuremath{\mathcal D}}
\newcommand{\RR}{\ensuremath{\mathcal R}}

\newcommand{\conv}{\ensuremath{\mathrm{conv}}}

\newcommand\makevec[1]{{\boldsymbol{#1}}}

\def \pp {\makevec{p}}
\def \qq {\makevec{q}}
\def \xx {\makevec{x}}
\def \yy {\makevec{y}}

\def \bY {\makevec{Y}}

\newcommand\pzero{\mbox{{\bf 0}}}
\newcommand\pone{\mbox{{\bf 1}}}
\newcommand\pinfty{\makevec{\infty}}

\DeclareMathOperator{\SEG}{SEG}
\DeclareMathOperator{\INTS}{INTS}
\DeclareMathOperator{\ETR}{ETR}
\DeclareMathOperator{\INEQ}{INEQ}
\DeclareMathOperator{\FEASIBLE}{FEASIBLE}
\DeclareMathOperator{\STRICTINEQ}{STRICT-INEQ}
\DeclareMathOperator{\RECOG}{RECOG}
\DeclareMathOperator{\CR}{cr}
\newcommand\RCR{\overline{\CR}}
\DeclareMathOperator{\CONV}{CONV}

\DeclareMathOperator{\sgn}{sgn}

\newcommand\eps{\varepsilon}

\newcommand{\heading}[1]{\vspace{1ex}\par\noindent{\bf\boldmath #1}}
\newcommand\defi[1]{{\bf\boldmath #1}}

\def\immediateFigure#1{%
\smallskip\begin{center}#1\end{center}\smallskip }

\newcommand{\immfig}[1]  
{\immediateFigure{\mbox{\includegraphics{#1}}}}

\newcommand{\immfigw}[2] 
{\immediateFigure{\mbox{\includegraphics[width=#2]{#1}}}}

\newlength{\fparwidth}
\newlength{\myparindent}
\begin{document}
\maketitle

\begin{abstract} 
A graph $G$ with vertex set $\{v_1,v_2,\ldots,v_n\}$
is an \emph{intersection graph of segments} if there are
segments $s_1,\ldots,s_n$ in the plane such that $s_i$ and $s_j$
have a common point if and only if $\{v_i,v_j\}$ is an edge of~$G$.
In this expository paper, we consider the
algorithmic problem of testing whether a given abstract graph
is an intersection graph of segments. 

It turned out that
this problem is complete for an interesting recently introduced
class of computational problems, denoted by $\exists\R$.
This class consists of problems that can be reduced, in polynomial
time, to solvability of a system of polynomial inequalities
in several variables over the reals. We 
discuss some subtleties in the definition of $\exists\R$, and
we provide a complete and streamlined account of a proof of the
$\exists\R$-completeness of the recognition problem for segment
intersection graphs. Along the way, we establish $\exists\R$-completeness
of several other problems. We also present a decision algorithm, due
to Muchnik, for the first-order theory of the reals.
\end{abstract}

\section{Introduction}

Let $G$ be a graph with vertex set $\{v_1,\ldots,v_n\}$.
We say that $G$ is an \emph{intersection graph of segments} if there are
straight segments $s_1,\ldots,s_n$ in the plane such that,
for every $i,j$, $1\le i<j\le n$, the segments $s_i$ and $s_j$
have a common point if and only if $\{v_i,v_j\}\in E(G)$.
Such segments $s_1,\ldots,s_n$ are called a \emph{segment
representation} of $G$.
For brevity, we will often say ``segment graph'' instead 
of ``intersection graph
of segments,'' and we let $\SEG$ denote the class of all segment
graphs.

Segment graphs constitute a difficult and much studied class.
They turned out to have strong algebraic aspects, and these are the focus
of the present paper. We should stress that there are also
numerous interesting non-algebraic results; one recent highlight,
settling a long-standing conjecture, is that every planar graph
belongs to $\SEG$ \cite{ChalopinGoncalves}, and another is a construction
of triangle-free $\SEG$ graphs with arbitrarily large chromatic 
number~\cite{pawlik2012triangle}.

The story we want to present here begins with the algorithmic question, 
given an abstract graph $G$, does it belong to $\SEG$? We will
refer to this as the \emph{recognition problem for segment graphs},
abbreviated $\RECOG(\SEG)$.

At first encounter, it is probably not obvious that 
there are any graphs at all
not belonging to $\SEG$. Here is one of the simplest examples,
a $K_5$ with subdivided edges:
\immfig{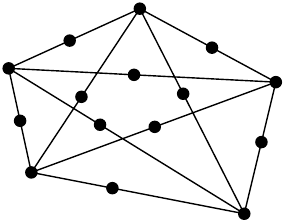}
Without going into any details, we just mention the idea of
the proof: assuming for contradiction that this graph is a segment graph,
one obtains a planar drawing of $K_5$, which is impossible.

After enough understanding of combinatorial properties of segment
graphs was accumulated, $\RECOG(\SEG)$ was proved NP-hard in
\cite{km-igs-94}. For many combinatorial problems, 
such as the existence of a Hamiltonian cycle in a graph, proving
NP-hardness is more or less the end of the story, since 
\emph{membership} in NP, i.e., a polynomial-time certificate
of a YES answer, is obvious. 

However, for $\RECOG(\SEG)$ this was only a beginning, since 
membership in NP is not clear at all, and today it
seems quite believable that $\RECOG(\SEG)$ is actually not
 in~NP.

Indeed, how should one certify that $G$ is a segment graph?
Some thought, which we leave to the reader, reveals that if
$G$ has a segment representation, then we may perturb and scale the segments
so that all of their endpoints have integer coordinates.
So if we knew that all the coordinates
of the endpoints have at most polynomially many digits  (in other
words, that the absolute value of the coordinates is bounded by
$2^{n^C}$, where $n=|V(G)|$ and $C$ is a constant), 
membership in NP would follow immediately. 

Serious people seriously conjectured that the number
of digits can be polynomially bounded---but it cannot.

\begin{theorem}\label{t:manydig}
For every sufficiently large $n$, there are $n$-vertex graphs in
$\SEG$  for which every segment representation with integral
endpoints has coordinates doubly exponential in $n$,
that is, with $2^{\Omega(n)}$ digits.
\end{theorem}

A first result of this kind, with
the weaker bound of $2^{\Omega(\sqrt n)}$, was proved in
\cite{km-igs-94}.
We will establish this weaker bound in Section~\ref{s:mnev}.
The stronger result in the theorem was obtained by
M\"uller and McDiarmid \cite{McDiarmidMuellerDisks}.
A similar ``large coordinates''
phenomenon in the setting of line arrangements, which will be mentioned later,
was observed earlier by Goodman, Pollack, and Sturmfels
\cite{gps-iscir-90}, based on a fundamental work of
Mn\"ev \cite{Mnev-in-Rochlin}.

Related to these developments, it was also gradually revealed
that the recognition problem for segment graphs
 is computationally at least as hard as
various other geometric and algebraic computational
problems---most notably, the solvability of a system of strict
polynomial inequalities in several variables over the reals.

More recently, Schaefer and \v{S}tefankovi\v{c} 
\cite{SchaeStef-Nash} (also see \cite{Schaefer-surv-exR}) introduced
 a new complexity class, denoted by $\exists\R$, and based on
previous work, they proved many natural and well-known problems
complete for this class, including $\RECOG(\SEG)$.
It is known that $\exists\R$ contains NP, and is contained in PSPACE
(the computational problems solvable in polynomial space), but 
for all we know, both of these inclusions may be strict.

The purpose of the present expository paper is twofold.
First,  we aim at a complete and streamlined
proof of the $\exists\R$-completeness of the recognition problem
for segment graphs, which contains many nice ideas and
has not been readily available in the literature in full.
Second, and more significantly, we use the problem $\RECOG(\SEG)$
as  an opportunity to introduce the class $\exists\R$
in detail and to treat some subtleties in its definition,
as well as several algebraic $\exists\R$-complete problems.
Hopefully this may contribute to  popularizing
this class, which so far may not be as widely known as it would
deserve. 

The best known algorithms for $\exists\R$-complete problems
work in exponential time and polynomial space, and they
are quite sophisticated and technically demanding.
We thus present another, suboptimal  but much simpler algorithm
due to Muchnik.
It actually solves a more general problem,
that of quantifier elimination in the first-order theory of the reals,
a classical problem first solved by Tarski in the 1950s.
Muchnik's algorithm can be explained in several pages, and it contains
some of the ideas used in more sophisticated algorithms.

Since this paper started out as a course material, it contains
a number of exercises. These may be useful for someone wishing
to teach the material, and also for self-study---ideas that
one discovers on his/her own are usually remembered much better.
Readers not interested in solving exercises may take them
just as a peculiar way of stating auxiliary results without proofs.

\heading{Acknowledgments. } 
I am especially grateful to Marcus Schaefer 
for answering numerous questions, reading a draft of my exposition,
and providing many useful comments.
I would  also like to thank to Rado Fulek, Vincent Kusters,
Jan Kyn\v{c}l, J\"urgen Richter-Gebert,
Zuzana Safernov\'a, and Daniel \v{S}tefankovi\v{c} for
various contributions,  such as proofreading, 
corrections, answering questions, etc. 
It was a pleasure to teach courses including
the material presented here together
with Pavel Valtr in Prague and with Michael Hoffmann and Emo Welzl
in Zurich.


\section{Recognizing segment graphs and the existential theory of~$\R$}

\heading{An algebraic re-formulation. }
It is not easy to find any algorithm at all to recognize segment graphs.
In desperation we can turn to computational algebra and
formulate the problem using polynomial equations and inequalities.

Suppose that a given graph $G=(V,E)$, $V=\{v_1,\ldots,v_n\}$,
is a segment graph, represented by segments $s_1,\ldots,s_n$.
By a suitable rotation of the coordinate system we can achieve that none
of the segments is vertical. Then the segment $s_i$ representing vertex $i$ can
be algebraically described as the set
$\{(x,y)\in\R^2: y=a_ix+b_i, c_i\le x\le d_i\}$ for some real
numbers $a_i,b_i,c_i,d_i$. 

Letting $\ell_i$ be the line containing
$s_i$, we note that $s_i\cap s_j\ne\emptyset$ if either 
$\ell_i=\ell_j$ and the intervals $[c_i,d_i]$ and $[c_j,d_j]$
overlap, or $\ell_i$ and $\ell_j$ intersect in a single point
whose $x$-coordinate lies in both of the intervals $[c_i,d_i]$ and 
$[c_j,d_j]$. As is easy to calculate, that $x$-coordinate
equals $\frac{b_j-b_i}{a_i-a_j}$. 

Let us introduce variables $A_i,B_i,C_i,D_i$ representing the unknown
quantities $a_i$, $b_i$, $c_i$, $d_i$, $i=1,2,\ldots,n$. (Following a convention
in a part of the algebraic literature, we will denote variables by capital
letters.) Then $s_i\cap s_j\ne\emptyset$ can be expressed by the following,
somewhat cumbersome, predicate $\INTS(A_i,B_i,C_i,D_i,A_j,B_j,C_j,D_j)$,
given by
\begin{eqnarray*}
&&\Big( A_i=A_j\wedge B_i=B_j \wedge\,\neg(D_i< C_j\vee D_j<C_i)\Big)\\
&&{}\vee \Big(A_i>A_j\wedge C_i(A_i-A_j)\le B_j-B_i\le D_i(A_i-A_j)\\
&& \ \ \ \ \ \ {}\wedge C_j(A_i-A_j)\le B_j-B_i\le D_j(A_i-A_j)\Big)\\
&& {}\vee \Big(A_i<A_j\wedge C_i(A_i-A_j)\ge B_j-B_i\ge D_i(A_i-A_j)\\
&& \ \ \ \ \ \ {}\wedge C_j(A_i-A_j)\ge B_j-B_i\ge D_j(A_i-A_j)\Big)\\
\end{eqnarray*}
(this is only correct if we ``globally'' assume that $C_i\le D_i$
for all $i$). Here, in accordance with a usual notation in logic,
 $\wedge$ stands for conjunction, $\vee$ for disjunction,
and $\neg$ for negation.

The existence of a $\SEG$-representation of $G$ can then be expressed
by the formula
\begin{eqnarray*}
&&(\exists A_1 B_1 C_1 D_1 \ldots A_n B_n C_n D_n)
\Bigl(\bigwedge_{i=1}^n C_i\le D_i\Bigr)\\
&&\ \ {}\wedge\biggl(\bigwedge_{\{i,j\}\in E} \INTS(A_i,B_i,B_i,D_i,A_j,B_j,C_j,D_j)\biggr)\\
&&\ \ {}\wedge \biggl(\bigwedge_{\{i,j\}\not\in E} \neg\INTS(A_i,B_i,B_i,D_i,A_j,B_j,C_j,D_j)\biggr).
\end{eqnarray*}
The existential quantifier in this formula quantifies all the 
variables written after it, and these variables range over $\R$,
the real numbers. 

\begin{exercise}\label{ex:pureseg}
Prove that every segment graph has a representation in which
no two segments lie on the same line (warning: this is not 
entirely easy). Therefore, the first part
of the predicate $\INTS$ is not really necessary.
\end{exercise}

\begin{exercise}\label{ex:conv-alg} Express the representability of
a given graph $G$ as an intersection graph of convex sets in the plane
 by a formula similar to the one above.
Hint: disjoint compact convex sets can be strictly separated by a line.
\end{exercise}

\heading{First-order theory of the reals. } 
It may seem that we have
replaced the problem of recognizing segment graphs, which was at
least geometrically intuitive, by a hopelessly complicated algebraic problem.
However, the \emph{only} known recognition algorithm relies exactly
on such a conversion and on a general algorithm for testing the validity
of this kind of formulas. What is more, recognizing segment graphs
turns out to be \emph{computationally equivalent} to validity
testing for a fairly broad class of formulas.

First we define a still more general class of formulas.
A \defi{formula of the first-order theory
of the reals} is made
of the constants $0$ and $1$, the binary operator symbols $+$, $-$ and
$\times$, the relation symbols $\le$, $<$, $\ge$, $>$, $=$, $\neq$,
the logical connectives
$\wedge$, $\vee$, $\neg$, and $\Leftrightarrow$ (equivalence),
 the quantifiers $\forall$ and $\exists$,
symbols for variables, and parentheses,
using simple syntactic rules, which the reader could no doubt easily supply.
In the rest of this paper, the word ``formula'' without other
qualifications means a formula in the first order theory of the reals.

Our including of ``derived'' symbols like $-$, $\neq$, $>$, $\Leftrightarrow$
 etc. is slightly
nonstandard, compared to the literature in logic or model theory, but
for considering computational complexity it appears more convenient
to regard them as primitive symbols.

Semantically, all variables range over the real 
numbers\footnote{In a first-order theory, we can quantify over
individual elements, in our case over the real numbers;
in a second-order theory, quantification over \emph{sets}
of real numbers would be allowed as well.} and
$+$, $-$, $\times$, $\le$, etc. are to be interpreted as the usual operations
and relations on~$\R$.

A formula is called a \defi{sentence} if it has no \emph{free} variables;
that is, every variable occurring in it is bound by a quantifier.
Our formula above expressing the $\SEG$-representability of a given
graph is an example of a sentence.

The first-order theory of the reals is sufficiently powerful
to express a wide variety of geometric problems. At the same
time, it is \emph{decidable}, which means that there
is an algorithm that, given a sentence, decides whether it
is true or false. We will present one such algorithm
in Section~\ref{s:muchnik}.
This is in contrast with the first-order
theories in many other areas of mathematics, such as the Peano arithmetic,
which are undecidable.

\heading{Prenex form. }
We will always assume, w.l.o.g., that the considered formulas
are in \defi{prenex form}; that is, of the form
$(Q_1X_1)(Q_2X_2)\ldots(Q_kX_k) \Phi$, where $Q_1,\ldots,Q_k$
are quantifiers and $\Phi$ is a quantifier-free formula.
An arbitrary formula is easily converted to a prenex form
by pushing the quantifiers outside using simple rules,
such as rewriting $\neg (\exists X)\Phi$ to $(\forall X)\neg\Phi$,
rewriting $((\exists X)\Phi)\vee \Psi$ to $(\exists X)(\Phi\vee\Psi)$
(this assumes that the name of the quantified variable $X$ does not
collide with the name of any variable in $\Psi$), etc.
Such a transformation changes the length of the formula
by at most a constant factor.

A general quantifier-free formula $\Phi$
with variables $X_1,\ldots,X_n$ has the form
\[
\Phi=\Phi(X_1,\ldots,X_n)= F(A_1,A_2,\ldots,A_m).
\]
Here $F$ is a Boolean formula and each $A_i$ is an
\emph{atomic formula} of the form $p_i(X_1,\ldots,X_n)~\mathrm{rel}~0$,
where rel is one of $<$, $\le$, $>$, $\ge$, $=$, $\neq$,  and
$p_i$ is a polynomial in $n$ variables with integer coefficients.

\heading{Size of formulas. }
When discussing the complexity of an algorithm accepting
a formula as an input, we will measure the input size
by the \emph{length} of the formula, i.e., the number of 
symbols in it.\footnote{The
length is basically proportional to the number of bits needed to encode
the formula. An exceptional case are the symbols of variables:
if the formula contains $n$ variables, then encoding 
a variable symbol needs about $\log n$ bits, at least on the average.
Thus, in the worst case, the encoding of a formula of length $L$
may need up to $L\log L$ bits.
But the complexity of the considered algorithms is so large
that this subtlety can be safely ignored, as is done in most of
the literature.}

In practice, one also writes integer constants, such as $13$, and powers, such as $X^{5}$, in the formulas.
An integer constant $k$ can be formed
with $O(\log k)$ symbols using the binary expansion of $k$;
e.g., $13=2^3+2^2+1=((1+1)+1)\times(1+1)\times(1+1)+1$.

On the other hand, a power $X^k$ is really an abbreviation 
for the $k$-fold product $X\times \cdots \times X$. 
We must be particularly careful with
more complicated expressions involving powers: for example,
$(\cdots ((X^2)^2\cdots)^2$, with $n$ second powers, could be written
using $O(n)$ symbols if we permitted the power operation---but we
do not allow it, and so this expression needs formula of length
about~$2^n$.

Thus, in brief, integer constants are written in binary, but exponents
must be taken as shorthands for multiple products.

Another thing to observe is that we cannot a priori assume polynomials
to be written in the usual form, as a sum of monomials with integer
coefficients. Indeed, the polynomial $(1+X_1)(1+X_2)\cdots(1+X_n)$
has $O(n)$ symbols according to our definition, but
if we multiplied out the parentheses, we would get $2^n$ distinct monomials,
and thus a formula of exponential length. This kind of subtleties can
usually be ignored in logic, but they start playing a role if we
consider polynomial-time reductions.

\heading{Complexity of decision algorithms. }
The best known decision algorithms for the first-order theory
of the reals can decide a sentence of length $L$ in time at most
$2^{2^{O(L)}}$, i.e., doubly exponential.

A more refined analysis of such algorithms uses several
other parameters besides the formula length, such as
the number of variables, the maximum degree of a polynomial
appearing in the formula, or the number of quantifier alternations
(e.g., the formula $(\exists X_1 X_2)(\forall X_3 X_4 X_5)(\exists X_6)\Phi$,
where $\Phi$ is quantifier-free, has three alternations).
For example, some of the algorithms are doubly exponential
only in the number of quantifier alternations, while the
dependence on other parameters is much milder.
We will discuss this somewhat more concretely at the end of 
Section~\ref{s:muchnik}.

\heading{Semialgebraic sets. } 
A set $S\subseteq \R^n$ is called \defi{semialgebraic}
if it can be described by a quantifier-free formula;
i.e., if $S=\{\xx\in\R^n:\Phi(\xx)\}$
for some $\Phi=\Phi(X_1,\ldots,X_n)$ as above. Thus, a semialgebraic
set is a set-theoretic combination of finitely many zero
sets and nonnegativity sets of polynomials.

\begin{exercise}\label{ex:1d-semi}{\rm (a)}
Describe semialgebraic sets in $\R^1$. 

{\rm(b)} Show that a semialgebraic
set in $\R^1$ defined by a quantifier-free formula of length $L$
has $O(L)$ connected components. More concretely, if the formula
$\Phi(X)$ involves polynomials $p_1(X),\ldots,p_m(X)$,
bound the number of components in terms of the degrees of the~$p_i$.
\end{exercise}

\heading{The existential theory of the reals and $\ETR$. }
Among the formulas of the first-order theory of the reals,
the quantifier-free ones are, in a sense, the simplest.
In particular, quantifier-free \emph{sentences} are trivially
decidable, since they involve only constants.

The next level of difficulty are the \emph{existential
formulas}, of the form 
\[
(\exists X_1 X_2\ldots X_k)\Phi,
\]
with $\Phi$ quantifier-free. For example, our formula above
for recognizing segment graphs is existential, and many
other questions about graph representations etc.\ can also be
expressed by an existential formula.

Let $\ETR$, abbreviating \emph{existential theory of the reals},
stand for the decision problem whose input is an existential
sentence $\Psi$ of the first-order theory of the reals, and the
output is the truth value of $\Psi$, TRUE or FALSE.

\begin{exercise}\label{ex:ETR-NPh}
 Show that $\ETR$ is {\rm NP}-hard.
\end{exercise}

The true complexity status of $\ETR$ is not known at present.
The best available algorithms show that it belongs to PSPACE,
i.e., it can be solved in polynomial space and exponential time
(this result is due to Canny \cite{Canny-pspace}; also see
\cite{BasuPollackRoy-book} for more recent algorithms).

We have seen that an instance of $\RECOG(\SEG)$ can be converted,
in polynomial time, into an equivalent instance of $\ETR$. This is usually
expressed by saying that $\RECOG(\SEG)$
\emph{reduces} to $\ETR$. A surprising and useful fact is
that, conversely, $\ETR$ reduces to $\RECOG(\SEG)$,
and thus these two problems are, up to polynomial-time
reductions, computationally equivalent. We will discuss this,
and related result, in the subsequent sections.

As an executive summary, one may remember that the first-order theory
of the reals can be decided in doubly-exponential time, and the
existential theory in singly-exponential time.

\heading{Remark: real-closed fields. }
Let us mention in passing that the first-order theory does not
fully describe the ordered field of the real numbers. Indeed,
there are other ordered fields, some of them quite interesting and
useful, that satisfy exactly the same set of sentences as the reals.
Such fields are called the \defi{real-closed fields}, and they
can be described by a simple axiomatic system (for which we refer
to the literature, e.g.,  \cite{BasuPollackRoy-book}).
In the literature, instead of the first-order theory of the reals,
one thus often speaks of the first-order theory of real-closed
fields.

One example of a real-closed field consists of all the \emph{algebraic}
real numbers. Another consists of all real \emph{Puiseux series};
these are formal series of fractional powers of the form
$\sum_{i=k}^\infty a_i \eps^{i/q}$, where $k\in\Z$, $q$ is a natural
number, $\eps$ is a formal variable, and the $a_i$
are real coefficients. 

Here we will not consider the Puiseux series in any detail. 
We just remark that they are
widely used  for \emph{perturbation arguments}, where one needs
to bring some semialgebraic sets into a suitably general position;
then the variable $\eps$ plays the role of an infinitesimal
quantity. A nice feature is that if we start with a real-closed
field and make the perturbation, we are again in a (larger) real-closed
field, and we can apply the same theory and
even the same algorithms, provided that
suitable computational primitives have been implemented.

\section{The complexity class $\exists\R$}

The class $\exists\R$ consists of the computational decision problems
that reduce to $\ETR$ in polynomial time. More precisely,
the reduction is the usual many-to-one reduction also used in the
definition of NP; that is, a decision problem $A$ many-to-one reduces
to a decision problem $B$ if there is a polynomial-time mapping $f$
such that $x$ is a YES-instance of $A$ iff $f(x)$ is a YES-instance of~$B$.


\heading{Restricted versions of $\ETR$ which are still complete
for $\exists\R$. }
The problem $\ETR$ is complete for $\exists\R$ by definition.
Here we will exhibit several restricted versions of $\ETR$ that
also turn out to be $\exists\R$-complete, although in some cases
this may be surprising. In the next section we will
discuss geometric $\exists\R$-complete problems. 

The significance of $\exists\R$-completeness is similar, on a smaller scale,
to that of NP-completeness: it bundles together many seemingly very different
problems into a single big complexity question.

The considered special cases of $\ETR$ all decide an existential
sentence 
\[(\exists X_1\ldots X_n)\Phi,
\]
 where $\Phi$ is
quantifier-free and of the following special forms:
\begin{itemize}
\item For the problem $\INEQ$, $\Phi$ is a conjunction of polynomial
equations and inequalities. Moreover, we require that the polynomials
in the equations and inequalities be written in the standard form,
as sums of monomials with integer coefficients. Thus, for example,
$(1+X_1)(1+X_2)\cdots(1+X_n)$ is not allowed.

\item $\STRICTINEQ$ is the special case of $\INEQ$
using only strict inequalities and no equations.

\item $\FEASIBLE$ is the special case of $\INEQ$
with a single polynomial equation $p(X_1,\ldots,X_n)=0$.
\end{itemize}

\begin{exercise} Formulate the problem of testing whether a given
graph belongs to $\SEG$ as an instance of $\STRICTINEQ$. 
Use Exercise~\ref{ex:pureseg} (which says that all segments
can be assumed to have distinct directions). 
\end{exercise}

We begin with showing that $\INEQ$, and even $\FEASIBLE$, are
$\exists\R$-complete. 

Even for $\INEQ$ this turns out to be trickier that one might think. 
One problem is with the exponential blow-up of formula size
when converting polynomials into the standard form, as was mentioned
in the preceding section. Another, similar problem appears when
dealing with an arbitrary Boolean formulas.

Fortunately, both of these problems can be solved by the same idea,
going back to Tseitin (in the context of Boolean satisfiability,
one speaks of the \defi{Tseitin transform} for converting 
Boolean formulas to \emph{conjunctive normal form}).
The main trick is to add new existentially quantified variables,
which are used to ``store'' the values of subformulas.

\begin{proposition}\label{p:feas-compl}
 The decision problems $\FEASIBLE$, and
consequently $\INEQ$, are $\exists\R$-complete.
In particular, solving a single polynomial equation in many variables over $\R$
is, up to a polynomial-time reduction,
as hard as solving a system of polynomial equations and inequalities
over~$\R$.
\end{proposition}

\begin{proof}
Let $(\exists X_1\ldots X_n)\Phi$ be an existential sentence
(in prenex form), of length $L$ and with $\Phi$ quantifier-free. 
We  transform it into
an equivalent instance of $\FEASIBLE$ of length $O(L)$.
The main trick is adding new real variables, which represent
the values of the various subformulas in~$\Phi$.

We consider the way of how $\Phi$ is recursively built from subformulas.
Here is an example, in which the subformulas are marked by lowercase Greek 
letters:
\[
\bigl(\underbrace{\underbrace{\underbrace{(\underbrace{X+Y}_{\vartheta})(\underbrace{Z-Y}_{\xi})}_{\zeta}>0}_{\alpha} 
\vee \underbrace{X\le 0}_{\beta}}_\gamma\bigr)\wedge
\underbrace{\neg(\underbrace{Y=0}_{\delta})}_{\eps}
\]
The subformulas $\vartheta,\xi,\zeta$ are \emph{arithmetic}, meaning that
they yield a real number, while $\alpha,\beta,\ldots$ are Boolean,
yielding true or false.

We process the subformulas one by one, starting from
the innermost ones and proceeding outwards. During this process, we 
build a new formula $\Gamma$, which is a conjunction of polynomial
equations. We start with $\Gamma$ empty, and we add equations to it
as we go along. We also add new existentially quantified real variables,
one or several per subformula.

In the example above, we start with the subformula $\vartheta$, we introduce
a new existentially quantified real variable $V_\vartheta$ 
(V abbreviating \emph{value}), and we add to $\Gamma$ the equation
$V_\vartheta=X+Y$. Similarly, for $\xi$ we add $V_\xi=Z-Y$, and
for $\zeta$ we add $V_\zeta=V_\vartheta V_\xi$. 
Thus, for each arithmetic subformula,
the corresponding variable  represents the value of the subformula.

For every Boolean subformula $\eta$, the plan is to introduce a corresponding 
real variable $W_\eta$, representing the truth value of $\eta$ in the
sense that $W_\eta=1$ means $\eta$ true and $W_\eta=0$ means $\eta$ false.

For our atomic formula $\alpha\equiv V_\zeta>0$, we begin with observing
that the strict inequality $V_\zeta>0$ is equivalent to
$(\exists S_\alpha)(V_\zeta S_\alpha^2=1)$, while its negation 
$V_\zeta\le 0$ is equivalent to $(\exists T_\alpha)(V_\zeta+T_\alpha^2=0)$.
Thus, with introducing the two auxiliary variables $S_\alpha$ and $T_\alpha$,
we obtain equations instead of inequalities. 

In order to have $W_\alpha$ represent the truth value of the 
subformula $\alpha$ as announced above, we would thus like to add to $\Gamma$
the  subformula
\begin{equation}\label{e:sadsubf}
(V_\zeta S_\alpha^2=1\wedge W_\alpha=1)\vee (V_\zeta+T_\alpha^2=0\wedge 
W_\alpha=0).
\end{equation}
Sadly, we cannot do that, since we are allowed to add only conjunctions
of equations, not disjunctions. But we observe that a conjunction
of equations $p=0\wedge q=0$ is equivalent to $p^2+q^2=0$, while a
disjunction $p=0\vee q=0$ can be replaced by $pq=0$. Thus, our desired
subformula \eqref{e:sadsubf} can be added in the following disguise:
\[
\big((V_\zeta S_\alpha^2-1)^2+(W_\alpha-1)^2\big)
\big( (V_\zeta+T_\alpha^2)^2+W_\alpha^2\big)=0
\]
Complicated as it may look, this equation still has
a constant length, even if we expand the polynomial
into the standard form.
This will be the case for every equation we add to~$\Gamma$.

The remaining two atomic formulas, $\beta$ and $\delta$, are handled analogously
and we leave this to the reader. Fortunately, the case of atomic formulas
is the most complicated one, and dealing with Boolean connectives is simpler.

The disjunction $\gamma\equiv\alpha\vee\beta$
is, by induction, equivalent to $W_\alpha=1\vee W_\beta=1$, and
the equation $W_\gamma=W_\alpha+W_\beta-W_\alpha W_\beta$ will do.
This, however, is a bit of a hack, relying on our particular representation
of the truth values. A more systematic approach is to first
represent the disjunction by $(W_\gamma=1\wedge (W_\alpha=1\vee W_\beta=1))
\vee (W_\gamma=0\wedge W_\alpha=0\wedge W_\beta=0)$, and then to
convert this formula into a single equation using the two tricks, 
$p^2+q^2$ and $pq$, as above.

The negation in the subformula $\eps$ is represented by $W_{\eps}=1-W_\delta$,
and for the conjunction in $\Phi\equiv\gamma\wedge\eps$
we can simply use $W_\Phi=W_\gamma W_\eps$.
Finally, we add the equation $W_\Phi=1$ to $\Gamma$, and this yields
the announced conjunction of equations $p_1=0\wedge\cdots \wedge p_m=0$
that is satisfiable iff the original formula $\Phi$ is. 

Since $\Phi$ has no more than $L$ subformulas, we have $m\le L$.
Each $p_i$ has length $O(1)$, and hence $\Gamma$ has length $O(L)$.
Finally, to get the desired instance of $\FEASIBLE$,
we replace $\Gamma$ with $p_1^2+\cdots+p^2_m=0$, which at most
doubles the formula length.
\end{proof}

\begin{exercise} In the above proof, we have not dealt with
a subformula of the form $\alpha\Leftrightarrow\beta$ (equivalence). 
Find a suitable
equation for replacing such a subformula.
\end{exercise}

%
%
%

\heading{Strict inequalities. }
Intuitively, the nature of $\STRICTINEQ$ seems to be different
from $\INEQ$. Indeed, if $\Phi(X_1,\ldots,X_n)$ is a conjunction
of strict inequalities, then the subset of $\R^n$ defined by it
is \emph{open}. Thus, for example, if it is nonempty, then it
always contains a point with rational coordinates; this, of course,
is not the case for equations, since, e.g., $X^2-2=0$ has only irrational
solutions.

Yet, as discovered in \cite{SchaeStef-Nash}, $\INEQ$
and $\STRICTINEQ$ are equivalent as computational problems.
For the proof we need a reasonably difficult
result of real algebraic geometry, of independent interest, which
we will take for granted.

\begin{theorem}\label{t:semibd} 
Let $S\subseteq\R^n$ be a semialgebraic
set defined by a quantifier-free formula of length $L$. If $S\ne\emptyset$,
then $S$ intersects the ball of radius
$R=2^{2^{CL\log L}}$ centered at $0$, were $C$ is a suitable absolute constant.
If $S$ is also bounded, then it is contained in that ball.
\end{theorem}

A more refined bound for $R$ is $2^{\tau (\Delta+1)^{O(n)}}$,
where $n$ is the number of variables in the formula,
$\Delta$ is the maximum degree of the polynomials in it, and $\tau$
is the maximum number of bits of a coefficient in the polynomials.
The ball of the stated radius even intersects every connected
component of $S$ and contains every bounded connected component.
The bound in the theorem is obtained by substituting the trivial
estimates $\tau\le L$, $\Delta\le L$, $n\le L$ into the refined bound.

This kind of result goes back to \cite[Lemma~9]{GrigorievVorobjov88}
(which deals with a special semialgebraic set, namely, the zero
set of a single polynomial), and the result as above about a ball intersecting
all connected components is \cite[Theorem~4.1.1]{BPR-QE}
(also see \cite[Theorem~13.14]{BasuPollackRoy-book}).
A statement directly implying the part with the ball containing
all bounded components is  \cite[Theorem~6.2]{BasuVorobjov-homot}.

\begin{proposition}\label{p:strii-compl}
 $\STRICTINEQ$ is $\exists\R$-complete. Thus, solving a system of strict 
polynomial inequalities over $\R$ is, up to a polynomial-time reduction,
as hard as solving arbitrary system of polynomial inequalities.
\end{proposition}

\begin{proof} We will reduce $\FEASIBLE$ to $\STRICTINEQ$.
Let the input formula for $\FEASIBLE$ ask for solvability of
$p(X_1,\ldots,X_n)=0$ and let its length be~$L$. 

By Theorem~\ref{t:semibd}, if $p(X_1,\ldots,X_n)=0$ is solvable,
then it has a solution $(x_1,\ldots,x_n)\in\R^n$
 with $\sum_{i=1}^n x_i^2<R^2$, where $R=2^{2^k}$ and $k=CL\log L$. 

We construct an instance of $\STRICTINEQ$ of the form
\[(\exists X_1\ldots X_n Y_1\ldots Y_k Z_1\ldots Z_\ell)\Phi,
\]
where $\ell=C_1L(\log L)^2$, with a sufficiently large constant $C_1$, and
$\Phi$ is the quantifier-free formula
\begin{eqnarray*}
&&Y_1> 0\wedge \cdots\wedge Y_k> 0\\
&&{}\wedge  Y_1<4 \wedge Y_2<Y_1^2\wedge\cdots\wedge Y_k<Y_{k-1}^2
\wedge X_1^2+\cdots+X_n^2<Y_k^2 \label{e:Phi12}\\
&&{}\wedge Z_1>4\wedge Z_2>Z_1^2\wedge\cdots\wedge Z_\ell>Z_{\ell-1}^2
\label{e:Phi3}\\
&&{}\wedge Z_\ell^2 p(X_1,\ldots,X_n)^2<1.\label{e:Phi4}
\end{eqnarray*}
The first two lines say that that
$\sum_{i=1}^n X_i^2<R^2$, and the last two lines
mean that $|p(X_1,\ldots,X_n)|<\delta:= 2^{-2^\ell}$.
The length of the formula is clearly bounded by a polynomial in $L$,
even if we convert all polynomials into the standard form.

We need to check that $\Phi$ is solvable iff $p(X_1,\ldots,X_n)=0$ is.
First, if $p(X_1,\ldots,X_n)=0$ is solvable, then there
is a solution $(x_1,\ldots,x_n)$ with $\sum_{i=1}^n x_i^2<R^2$,
and it can be extended to a solution of $\Phi$.

Conversely,
let us suppose that $p$ has no zero. Let us consider the semialgebraic
set  $S\subseteq\R^{n+k+1}$ given by the following formula 
$\Xi=
\Xi(X_1,\ldots,X_n$, $Y_1,\ldots,Y_k,Z)$:
\begin{eqnarray*}
&&Y_1> 0\wedge\cdots\wedge Y_k> 0\\
&& {}\wedge Y_1<4 \wedge Y_2<Y_1^2\wedge\cdots\wedge Y_k<Y_{k-1}^2\wedge X_1^2+\cdots+X_n^2<Y_k^2\\
&&{}\wedge Z^2 p(X_1,\ldots,X_n)^2 = 1.
\end{eqnarray*}
The length of $\Xi$ is  $K=O(k)$.
Since we assume that $p$ has no zeros, $S$ is a bounded set,
and hence by Theorem~\ref{t:semibd}, it is contained in the ball
of radius $2^{2^{CK\log K}}\le
2^{2^\ell}=1/\delta$, and in particular,
$|Z|\le 1/\delta$ in every solution. This implies that
for every $(x_1,\ldots,x_n)\in\R^n$ with $\sum_{i=1}^nx_i^2<R^2$,
we have $|p(x_1,\ldots,x_n)|\ge\delta$, and hence $\Phi$ has no
solution.
\end{proof}

\begin{exercise} 
Find a bivariate polynomial $p(X,Y)$ with integer
coefficients with $p(x,y)>0$ for all $(x,y)\in\R^2$ and with
$\inf_{\R^2} p(x,y)=0$.
\end{exercise}

\heading{Remark: size of segment representations and 
Oleinik--Petrovski\v{i}--Milnor--Thom. }
By reasoning very similar to the proof of Theorem~\ref{t:semibd},
it can also be shown that if $S$ is a nonempty semialgebraic set
defined by a formula of length $L$ that is a conjunction
of \emph{strict} inequalities, then $S$ contains a rational point
 whose coordinates are fractions of integers with $2^{O(L\log L)}$ bits
(a more refined bound in the notation below Theorem~\ref{t:semibd}
is $\tau\Delta^{O(n)}$; see \cite[Theorem~13.15]{BasuPollackRoy-book}). 

Using the refined bound on the ETR formula 
expressing $\RECOG(\SEG)$ given at the beginning of this chapter, 
we get that every segment graph on $n$ vertices has a segment representation 
in which the endpoint coordinates have $2^{O(n)}$ digits
(there are $4n$ variables, the polynomials are at most quadratic,
and $\tau=O(1)$). Hence Theorem~\ref{t:manydig} is tight up to
a multiplicative constant in the exponent.

 Theorem~\ref{t:semibd} is also  related to a result of
real algebraic geometry that proved truly fundamental for discrete
geometry, theoretical computer science, and other fields:
the Oleinik--Petrovski\v{i}--Milnor--Thom theorem (also often
called the Milnor--Thom theorem or Warren's theorem in the literature).
This is a result that bounds the maximum number of connected components
of semialgebraic sets. In a modern version, it can be stated
as follows: 

\begin{theorem}
Let $p_i(X_1,\ldots,X_n)$ be polynomials
of degree at most $\Delta$, $i=1,2\ldots,m$, and for every
sign vector $\sigma\in\{-1,0,+1\}^m$ let $S_\sigma\subseteq\R^n$
be defined as
\[ \Bigl\{(x_1,\ldots,x_n)\in\R^n:
\bigwedge_{i=1}^m \sgn p_i(x_1,\ldots,x_n)=\sigma_i\Bigr\}.
\]
 Then
for $m\ge n\ge 2$,
\[
\sum_{\sigma\in\{-1,0,+1\}^m}\#S_\sigma \le \left(\frac{50\Delta m}{n}\right)^{n},
\]
where $\#S$ denotes the number of connected components of~$S$.
\end{theorem}

\begin{exercise} Apply the theorem just stated to show that
there are at most $2^{O(n\log n)}$ nonisomorphic segment graphs 
on $n$ vertices.
\end{exercise}

Applications of the Oleinik--Petrovski\v{i}--Milnor--Thom theorem abound in
the literature, from simple ones as in the exercise to sophisticated
uses, and it is useful to be aware of this kind of result.

\section{Stretchability and the Mn\"ev universality theorem}\label{s:mnev}

Here is the main goal of this section.

\begin{theorem}\label{t:seg-compl} 
The problem $\RECOG(\SEG)$, recognizing
segment graphs, is $\exists\R$-complete.
\end{theorem}

As a by-product of the proof, we will also get a proof
of a weaker version of Theorem~\ref{t:manydig} (exponentially many
digits needed for a $\SEG$ representation), with $2^{\Omega(\sqrt n)}$ 
digits instead of $2^{\Omega(n)}$.

\heading{A sample of other geometric $\exists\R$-complete problems.}
Before starting our development, we list several other 
$\exists\R$-complete problems; many more such problems can be found
in the literature (see, e.g., \cite{SchaeStef-Nash,Schaefer-surv-exR}).
\begin{itemize}
\item Recognition of intersection graphs
of \emph{unit disks} in the plane \cite{DBLP:journals/dcg/KangM12}.
\item Recognition of $\CONV$, intersection graphs of convex sets in $\R^2$
\cite{Schaefer-surv-exR}.
\item Determining the \emph{rectilinear crossing number}
$\RCR(G)$ of a graph \cite{DBLP:journals/dcg/Bienstock91},
that is, the minimum possible number of edge crossings in
a drawing of $G$ in the plane, with edges drawn as straight
segments.
\item (The \emph{Steinitz problem}) 
Given a partially ordered set, determining whether
it is isomorphic to the set of all faces of a convex polytope 
ordered by inclusion (this partially ordered set
is called the  \emph{face lattice} of the polytope); see
\cite[Corollary~9.5.11]{Bjorner:Oriented}. 
\end{itemize}

\subsection{From segment graphs to line arrangements}

\heading{STRETCHABILITY. } 
Let us consider a set $L$ of $n$ lines in the plane.
The \defi{arrangement} of $L$ is the partition of the plane
into convex subsets induced by $L$. Four of such subsets are marked
in the next picture: a \emph{vertex}, which is an intersection of two lines,
two \emph{edges}, which are pieces of the lines delimited by the vertices,
and a \emph{region}, which is one of the pieces obtained after cutting the plane
along the lines.
\immfig{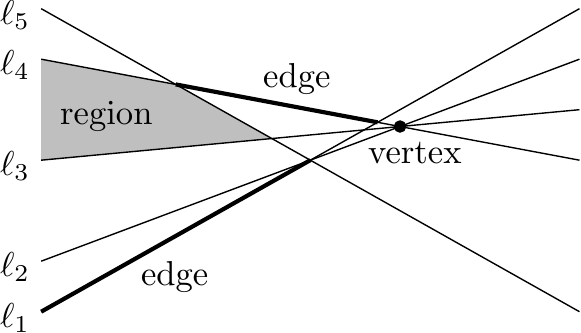}

We want to define a \emph{combinatorial description} of such a line arrangement,
the idea being that two arrangements with the same description look 
combinatorially the same. Such descriptions are
systematically studied in the theory of \emph{oriented matroids}
(see, e.g., \cite{Bjorner:Oriented}). 
There are several reasonable definitions of a combinatorial description,
all of them essentially equivalent.
We will use one that is simple and convenient for our present purposes, 
although not among the most conceptual ones.

For simplicity, let us assume that none of the lines of $L$ is vertical
and every two intersect. Then we number the lines $\ell_1,\ldots,\ell_n$
in the order of decreasing slopes, as in the above picture,
and for each $i=1,2,\ldots,n$, we write down the numbers of the lines
intersecting $\ell_i$ as we go from left to right. From the picture
we thus get the following five lists:
\[
({\scriptstyle{2\atop 5}},3,4),
({\scriptstyle{1\atop 5}}, {\scriptstyle{3\atop4}}),
(5,1,{\scriptstyle{2\atop 4}}),
(5,1,{\scriptstyle{2\atop 3}})
(4,3,{\scriptstyle{1\atop 2}}).
\]
Several numbers in a column mean that the corresponding intersections
coincide.
These $n$ lists constitute (our way of)
the combinatorial description of the arrangement of~$L$.

The decision problem \defi{STRETCHABILITY} can now be defined as follows: 
given $n$ lists of integers, 
with some of the entries arranged in columns as above, 
decide whether they constitute a combinatorial description
of an arrangement of $n$ lines. 

There are some obvious consistency conditions on the lists: for example,
the $i$th list should contain each element of $[n]\setminus\{i\}$
exactly once, and for $i<j<k$, if $\ell_i$ intersects $\ell_k$ before 
$\ell_j$, then $\ell_j$ must intersect $\ell_i$ before $\ell_k$,
etc. It is not hard to formulate conditions on the lists
so that they provide a combinatorial description of an arrangement
of \defi{pseudolines}. Here a set of $n$ curves is called
a set of (affine) pseudolines if every curve intersects every 
vertical line exactly once, and every two curves cross exactly once.

The real problem in STRETCHABILITY is in recognizing whether 
a combinatorial description of an arrangement of pseudolines
also fits to an arrangement of lines, i.e., if the pseudolines
can be ``stretched''---this is where the name comes from.
The following is a famous example of a non-stretchable arrangement:
\immfig{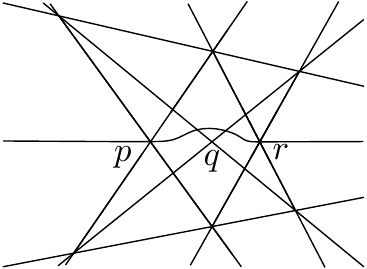}
In this picture, all of the pseudolines are straight except for one
(and only a part of the arrangement not including all of the intersections
is shown, for space reasons). The non-stretchability relies on
an ancient theorem of Pappus, which asserts that if eight straight lines
intersect as indicated, then the line passing through $p$ and $r$
also has to contain~$q$.

An arrangement of lines is called \defi{simple} (or, sometimes,
\emph{uniform}) if no three lines have a point in common. 
The decision problem \defi{SIMPLE STRETCHABILITY}, a special case
of STRETCHABILITY, asks whether
given lists form the combinatorial description of a simple
arrangement of lines. 

In due time, we will see that $\STRICTINEQ$ reduces to 
SIMPLE STRETCHABILITY, and thus SIMPLE STRETCHABILITY is $\exists\R$-complete.
But here we deal with the following result from
\cite{km-igs-94}, the first step in the proof of 
Theorem~\ref{t:seg-compl}:

\begin{proposition} {\rm SIMPLE STRETCHABILITY} reduces to $\RECOG(\SEG)$;
i.e., recognizing segment graphs is no easier than testing 
simple stretchability.
\end{proposition}

\begin{proofhd}{Sketch of proof}
First we explain a construction of a segment representation of a suitable
graph from a given simple line arrangement.
We begin with adding a vertical segment $v$ intersecting
all the lines and lying left of all intersections in the arrangement.
We also add two other segments that together with $v$ enclose all of the
intersections in a triangle $T$. Then we shorten all of the lines
to segments ending outside~$T$:
\immfig{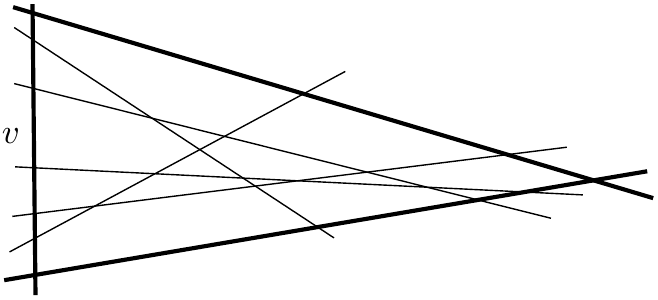}

Next, in a small neighborhood of each of the segments in this picture, 
including the newly added ones, we add an \emph{ordering gadget},
made of many segments, as is indicated below:
\immfig{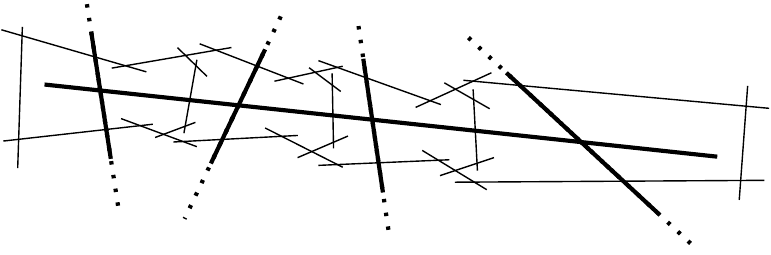}
The original segments are drawn thick. Let $G$ be the intersection
graph of the resulting set of segments.

We have presented the construction geometrically, but it can obviously be
``combinatorialized'' so that, given a combinatorial description,
which may or may not correspond to an actual line arrangement,
it produces a graph $G$. Clearly, if the description does correspond
to a line arrangement, then $G$ has a segment representation,
and it remains to prove the converse---if $G$ has a segment representation,
then there exists an arrangement with the given description. 

The key step in the proof is in showing that the ordering gadget
indeed forces the appropriate linear ordering of the intersections
along each of the original segments, up to a reversal. 
This is quite intuitive,
and for a detailed proof we refer to \cite{Schaefer-surv-exR}.

Once we know this, we can discard the segments of the ordering gadgets.
We make an affine transformation of the plane so that the segment $v$ is vertical
and the triangle $T$ is to the right of it. 
Then all intersections of the remaining
segments must be inside $T$, and these segments can be extended
to full lines. Then either the resulting line arrangement or its
upside-down mirror reflection conform to the given combinatorial
description.
\end{proofhd}

\subsection{From line arrangements to point configurations}

Our main goal now is showing $\exists\R$-completeness of SIMPLE
STRETCHABILITY. First we pass from line arrangements to
point configurations; these two settings are equivalent,
but point configurations appear more convenient for the subsequent
development. The passage relies on \emph{line-point duality},
a basic concept of projective geometry, and the material in this
section may be rather standard/easy for many readers.

For combinatorially describing a configuration of points in the plane,
we use the notion of \emph{order type} (which, under the name
of a \emph{chirotope}, is also one of the possible axiomatizations
of oriented matroids; see \cite{Bjorner:Oriented}). 

The order type
is defined for a \emph{sequence} $(p_1,p_2,\ldots,p_n)$ of points
in $\R^2$ (or in $\R^d$, but we will not need that). For an ordered
triple $(p,q,r)$ of points, we define the \emph{sign} depending
on the direction, left/straight/right,
 in which we turn when going from $p$ to $q$
and then to $r$:
\immfig{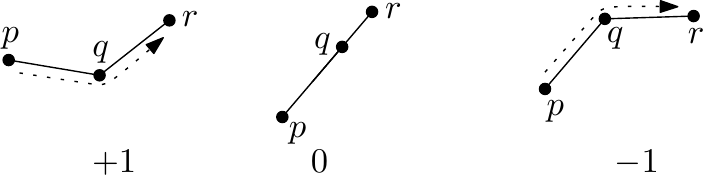}
The sign can  also be defined algebraically,
 as the sign of the determinant of the matrix with rows $q-p$ and $r-p$.
We define a \emph{combinatorial order type} as a mapping
${[n]\choose 3}\to \{-1,0,+1\}$, and the
 \defi{order type} of a sequence $(p_1,p_2,\ldots,p_n)$ is the mapping
assigning to every triple $i<j<k$
the sign of $(p_i,p_j,p_k)$.

The \defi{ORDER-TYPE REALIZABILITY} problem has a combinatorial order type as
the input, and asks if it is the order type of an actual point sequence.
We also have the \defi{SIMPLE ORDER-TYPE REALIZABILITY} variant, 
where all the signs are $\pm 1$'s, or in other words, the
considered point set is in general position, with no collinear
triples. 

\begin{lemma}\label{l:otr}
{\rm SIMPLE ORDER-TYPE REALIZABILITY} reduces to {\rm SIMPLE STRETCHABILITY}.
\end{lemma}

\begin{proof}
As announced above, the proof is a simple
use of a line-point duality. In view of our way of combinatorially
describing line arrangements, we use the following version of duality:
To a point $p=(a,b)\in \R^2$, we assign the (nonvertical) line $\DD(p)$
with equation $y=ax-b$. Conversely, to every nonvertical line $\ell$,
which has a unique equation of the form $y=ax-b$, we assign the point 
$\DD(\ell)=(a,b)$. The basic property of $\DD$, quite easy to check
from the definition, is that a point
$p$ lies below, on, or above a nonvertical line $\ell$ iff
the dual point $\DD(\ell)$ lies below, on, or above the dual line $\DD(p)$.
In particular, $\DD$ preserves line-point incidences and thus maps
the intersection of two lines to the line spanned by the two dual points,
and vice versa.

Let $\tau$ be a simple combinatorial order type.
Supposing that there is a point sequence $(p_1,\ldots,p_n)$ realizing
it, we may assume that the point $p_n$ lies on the negative $y$-axis
very far below all the other points (``at $-\infty$'' for all practical
purposes). This can be achieved by a suitable \emph{projective transformation} 
of the plane; for readers not familiar with this concept we provide 
an express introduction
below.

With this assumption, it is easy to see that
we can read off the ordering of the $x$-coordinates
of $p_1$ through $p_{n-1}$ from the signs of the triples involving $p_n$.
Thus, if we consider the dual lines $\DD(p_1),\ldots,\DD(p_{n-1})$,
we know the ordering of their slopes.

We can also reconstruct the combinatorial description of the line arrangement
from the order type. For example, if $i<j<k$ and we want to know which
of the lines $\ell_i$, $\ell_j$ intersects $\ell_k$ first, it suffices
to know whether the intersection $\ell_i\cap\ell_j$ lies above or below
$\ell_k$. This in turn is equivalent to the dual point $\DD(\ell_k)$
lying above or below the line through $\DD(\ell_i)$ and $\DD(\ell_j)$,
by the properties of the duality mentioned above, and this last piece of
information is contained in the order type. A similar reasoning works
for other orderings of the indices $i,j,k$.

Conversely, the order type of $(p_1,\ldots,p_n)$ can be reconstructed
from the combinatorial description of the arrangement of 
the dual lines $\DD(p_1),\ldots,\DD(p_{n-1})$.

Summarizing, from the combinatorial order type $\tau$ we can construct, in polynomial time of course, a combinatorial description of an arrangement of $n-1$
lines that is stretchable iff $\tau$ is realizable.
\end{proof}

\heading{Projective transformations. } Above and in the sequel we need
some properties of projective transformations. For readers not familiar with
this topic we provide several introductory sentences, referring, e.g., to
Richter-Gebert \cite{richter2011perspectives} for a solid introduction
to projective geometry.

A projective transformation  
maps the plane onto itself, or more precisely, it maps the plane
minus an exceptional line onto the plane minus another exceptional line.
Here is a quite intuitive geometric way of thinking of projective 
transformations. We fix two planes $\pi$ and $\sigma$ in $\R^3$
and identify each of them with $\R^2$ by choosing a system of 
Cartesian coordinates in it. We also choose a point $o\not\in \pi\cup\sigma$,
and the projective transformation is obtained by projecting $\pi$ from $o$
into $\sigma$.

\immfig{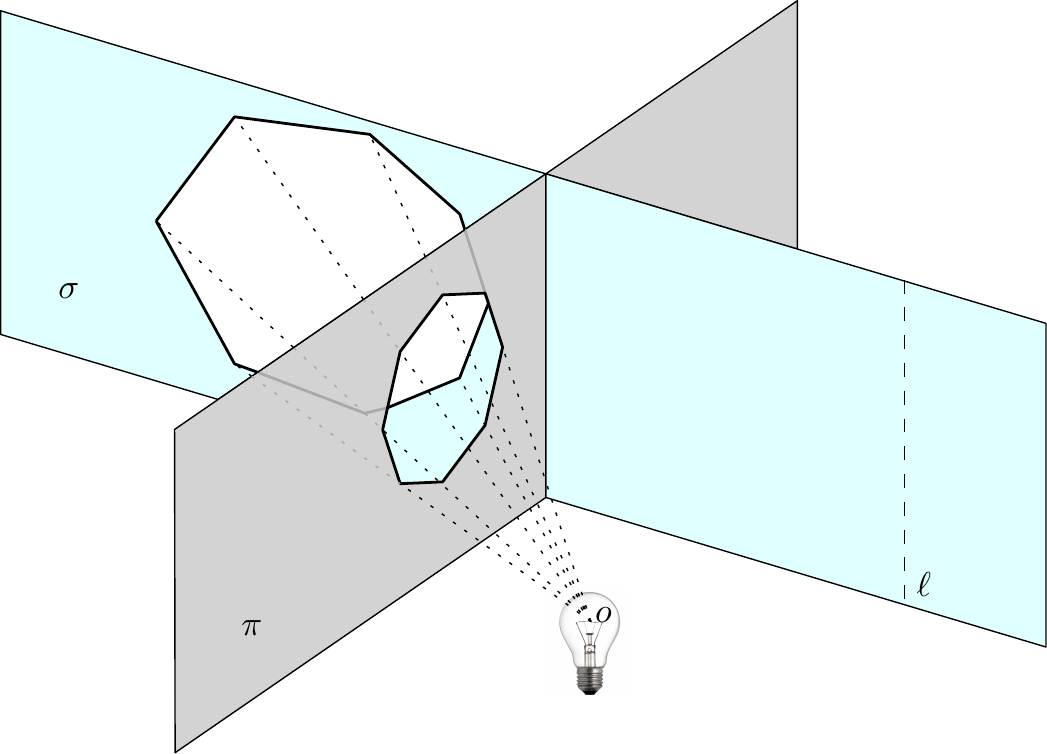}

The exceptional line in $\pi$ is labeled by $\ell$ in the picture;
it is the intersection of $\pi$ with the plane
through $o$ parallel to $\sigma$. In projective geometry, one completes
the plane by a line at infinity, to the \emph{projective plane},
and projective transformations map the projective planes bijectively
onto itself. In particular, the exceptional line $\ell$ is mapped
onto the line at infinity. 

A projective transform maps straight lines to straight lines and
preserves line-point incidences. It also preserves the order type,
provided that all of the points lie on the same side of the 
exceptional line.


\subsection{General position and constructible configurations}

So far we are carrying along the simplicity/general position requirement,
which came naturally from the setting of segment graphs---there we cannot
force three segments to meet at a single point.
For reducing the solvability of strict polynomial inequalities, $\STRICTINEQ$,
to a realizability problem for point configurations, the general position
requirement would be very inconvenient. Here we replace it with
a considerably weaker requirement, called constructibility, by means of
an  ingenious  trick, first used by  
Las Vergnas (see \cite[Prop.~8.6.3]{Bjorner:Oriented}).

We say that a point sequence $(p_1,p_2,\ldots,p_n)$ is 
\defi{constructible} if, possibly after renumbering the points
suitably, the following hold:
\begin{itemize}
\item No three among $p_1,\ldots,p_4$ are collinear (this is usually
expressed by saying that $p_1,\ldots,p_4$ form a \emph{projective basis}).
\item Each $p_i$, $i>4$, lies on at most two of the lines spanned
by $p_1,\ldots,p_{i-1}$.
\end{itemize}
In particular, every sequence in general position is constructible.

The notion of constructibility makes sense for a combinatorial
order type, realizable or not. In the corresponding
algorithmic problem  \defi{CONSTRUCTIBLE ORDER-TYPE REALIZABILITY},
the input is a combinatorial order type for which the points
are already ordered as in the definition of constructibility
(so that we need not worry about finding the right ordering
in polynomial time).

Intuitively, a constructible point configuration is one that can be constructed
from the initial four points using only a ruler, i.e., by passing lines
through pairs of already constructed points and placing each new point
either arbitrarily, or arbitrarily on an already constructed line, or to
the intersection of two already constructed lines. Such a ruler construction, 
though, only takes into account which triples of points should be collinear,
and not the signs of triples. Thus,  a constructible combinatorial order
type may still be unrealizable.

\begin{proposition}\label{p:constr-simp}
 {\rm CONSTRUCTIBLE ORDER-TYPE REALIZABILITY} 
 is reducible to 
{\rm SIMPLE ORDER-TYPE REALIZABILITY}.
\end{proposition}

\heading{The construction: first version. } To prove the proposition, 
given a constructible combinatorial order type $\tau$, 
we want to construct, in polynomial time, a simple combinatorial
order type $\xi$ such that $\xi$ is realizable iff $\tau$ is.

In order to make the proof more accessible, we first explain a somewhat
simplified version of the construction in terms of specific point
sequences, i.e., assuming that we are given a sequence $\pp$
realizing the given order type~$\tau$. 

Thus, let $\pp=(p_1,\ldots,p_n)$ be a constructible sequence
with order type $\tau$, with a numbering as in the definition of 
constructibility. We inductively construct sequences 
$\qq^{(n)}=\pp,\qq^{(n-1)},\ldots,\qq^{(4)}$. The first
$t$ points of $\qq^{(t)}$ are $p_1$ through $p_t$.  The
final product of this construction is $\qq=\qq^{(4)}$, which is
a sequence of fewer than $3n$ points
in general position.

For obtaining $\qq^{(t-1)}$ from $\qq^{(t)}$, we distinguish three
cases, depending on the number of the lines spanned by $p_1,\ldots,p_{t-1}$
that pass through~$p_t$ (there are at most two by constructibility):
\begin{enumerate} 
\item[(2)]  First suppose that $p_t$ lies on two such lines;
let they be the lines $p_ip_j$ and $p_kp_\ell$. They divide
the plane into four sectors. The sequence $\qq^{(t-1)}$  is obtained from  $\qq^{(t)}$
by replacing $p_t$ with four points $p_{t,1},\ldots,p_{t,4}$,
one in each of the sectors, lying very close
to $p_t$:
\immfig{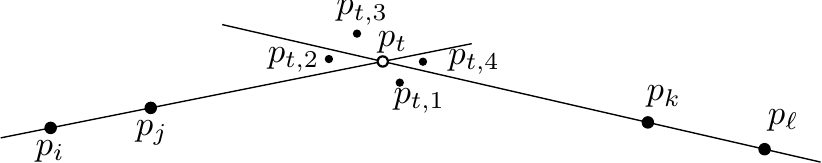}
\item[(1)] If $p_t$ lies on a single line $p_ip_j$, we replace it
with three points $p_{t,1},p_{t,2},p_{t,3}$:
\immfig{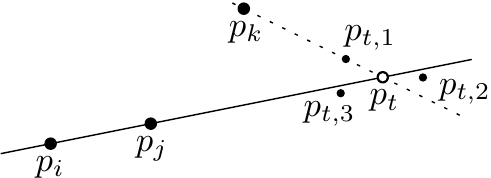}
We first choose another point $p_k$, $k<t$, not lying on $p_ip_j$,
we place $p_{t,1}$ on the line $p_tp_k$, and then $p_{t,2}$ and $p_{t,3}$
are on the opposite side of $p_ip_j$ and on different sides of $p_tp_k$.
This makes sure that $p_t\in\conv(p_{t,1},p_{t,2},p_{t,3})$.
\item[(0)] If $p_t$ does not lie on any line spanned by $p_1$ through $p_{t-1}$,
then we simply set $\qq^{(t-1)}=\qq^{(t)}$.
\end{enumerate}
In both of the cases (2) and (1), we place the new points so that
they do not lie on any line spanned by the points 
 of~$\qq^{(t)}$ minus $p_t$. 
We also place them so close to $p_t$
that every line spanned by the points
of $\qq^{(t)}$  and not passing through $p_t$ has
them on the same side as~$p_t$. 

By induction, the only collinear triples in $\qq^{(t)}$
are those spanned by $p_1,\ldots,p_{t}$, and in particular,
$\qq=\qq^{(4)}$ is in general position. Here is the main property
of the construction.

\begin{lemma}\label{l:otrr} Given any sequence $\tilde\qq$ realizing the
order type $\xi$ of $\qq$, one can construct a sequence $\tilde\pp$
realizing the order type $\tau$ of the original sequence~$\pp$. 
\end{lemma}

\begin{proof}
It suffices to check that if $\tilde\qq^{(t-1)}$ has 
the same order type as $\qq^{(t-1)}$,
 we can obtain $\tilde\qq^{(t)}$ with the same order type
as $\qq^{(t)}$. Let us assume that
$\qq^{(t-1)}$ was constructed from $\qq^{(t)}$ according to
 case (2), since case (1) is analogous and case (0) trivial.

It is clear how the desired $\tilde\qq^{(t)}$ should be
obtained from $\tilde\qq^{(t-1)}$:
by deleting the points $\tilde p_{t,1},\ldots,\tilde p_{t,4}$ and placing
$\tilde p_t$ to the intersection of the two lines $\tilde p_i \tilde p_j$
 and $\tilde p_k \tilde p_\ell$.
What needs to be checked is that the resulting
$\tilde\qq^{(t)}$ has the order type of $\qq^{(t)}$.

We need to consider only the signs of the triples involving $\tilde p_t$.
Let $\tilde a$ and $\tilde b$ be the other two points in such a triple,
and let $a$ and $b$ be the corresponding points in $\qq^{(t)}$.

The points $\tilde p_{t,1},\ldots,\tilde p_{t,4}$ form a convex
quadrilateral containing the intersection of the lines
$\tilde p_i \tilde p_j$ and $\tilde p_k \tilde p_\ell$,
since this information is specified by the order type of~$\qq^{(t-1)}$.

Thus, if the line $\tilde\ell$ spanned by $\tilde a$ and $\tilde b$ 
avoids the quadrilateral, then $\tilde p_t$, lying inside
the quadrilateral, is on the same side of $\tilde\ell$
as $\tilde p_{t,1},\ldots,\tilde p_{t,4}$, and so
the sign of the triple $(\tilde a,\tilde b,\tilde p_t)$ is 
the same as for the corresponding triple $(a,b,p_t)$ in~$\qq^{(t)}$.

If $\tilde\ell$ does intersect the quadrilateral, then the 
corresponding line for $\qq^{(t-1)}$ also intersects the corresponding
quadrilateral, and since $p_{t,1},\ldots,p_{t,4}$ were placed
sufficiently close to $p_t$,  the points $a,b,p_t$ must be
collinear. But the only collinear triples
in $\qq^{(t)}$ involving $p_t$ 
lie on the two lines $p_ip_j$ and $p_kp_\ell$ defining $p_t$.
Hence $\tilde a$ and $\tilde b$ lie on the corresponding line
for $\tilde\qq^{(t)}$ and form a collinear triple with $\tilde p_t$
as well.
\end{proof}

\heading{Lexicographic extensions. } 
It may seem that Proposition~\ref{p:constr-simp} is already proved, but 
there is still a problem we need to address. Namely, we need a construction
phrased solely in terms
of order types; that is, we want to construct the order type
$\tau^{(t)}$ of $\qq^{(t)}$ directly from the order type
of $\tau^{(t-1)}$ of $\qq^{(t-1)}$, without relying on a particular
realization of $\tau^{(t-1)}$. Indeed, the construction must also make
sense for non-realizable combinatorial order types.

The construction as presented above is not yet suitable for this purpose:
when choosing the new points $p_{t,1},p_{t,2},\ldots$, we have not specified
the order type fully, since the position of a line spanned by 
two new points among the old points is not determined, and similarly
for lines spanned by a new point and an old point. We thus need to
be more specific, and for the definition, we use the following general notion.

Let $\xx=(x_1,\ldots,x_m)$ be an arbitrary point sequence in $\R^2$, 
let $x_i,x_j,x_k$ be three non-collinear points of $\xx$, and let
$\rho,\sigma\in\{-1,+1\}$ be signs. For every $\eps>0$, we
consider the point $x=x(\eps)= x_i+\rho\eps(x_j-x_i)+ \sigma\eps^2(x_k-x_i)$.
That is, we move $x_i$ a bit towards $x_j$ (if $\rho=+1$)
or away from it, and an even much smaller
bit towards $x_k$ or away from it. Let us form a new
sequence of $m+1$ points by inserting $x(\eps)$ into $\xx$
after $x_i$. It is not hard to see that
for all sufficiently small $\eps>0$, the order type of
the new sequence is the same,
and crucially, it can be figured out from the order type of~$\xx$.
We call this new sequence a \defi{lexicographic extension} of $\xx$,
and we write the new point $x$ as $[x_i,x_j^\rho,x_k^\sigma]$.
Analogously we define $[x_i,x_j^\rho]$, with just one move.

\begin{proofof}{Proposition~\ref{p:constr-simp}}
Given a constructible combinatorial order type $\tau$, we
produce the simple combinatorial order type $\xi$ essentially
according to the construction above, only we specify the 
way of adding the new points in terms of lexicographic extensions.

Namely, the four new points in case (2) of the construction above
are obtained by four successive lexicographic extensions of
the current sequence, as
$p_{t,1}:= [p_t,p_i^+,p_k^+]$, $p_{t,2}:=[p_t,p_i^+,p_k^-]$,
$p_{t,3}:=[p_t,p_i^-,p_k^-]$, and $p_{t,4}:=[p_t,p_i^-,p_k^+]$.
Thus, a somewhat more realistic illustration to case (2) is this:
\immfig{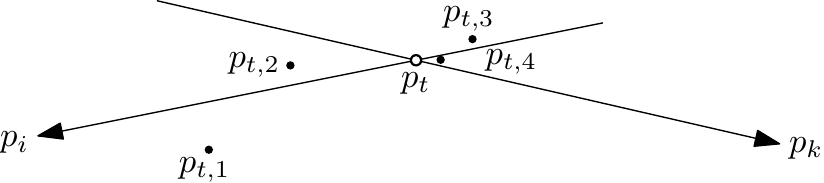}
No illustration can be quite realistic, since the $\eps$'s should
actually be very small and decrease very fast with the successive
lexicographic extensions.

Similarly, in case (1), where $p_t$ lies on the line $p_ip_j$
and $p_k$ is another point not lying on that line,
we set
$p_{t,1}:=[p_t,p_k^+]$, 
$p_{t,2}:=[p_t,p_i^+,p_k^-]$, and $p_{t,3}:=[p_t,p_i^-,p_k^-]$.
Now the construction is fully specified in terms of order types.

If $\tau$ is realizable, then, clearly, $\xi$ is realizable (we just
perform the construction geometrically). If $\xi$ is realizable,
then $\tau$ is realizable by Lemma~\ref{l:otrr}.
\end{proofof}

\subsection{The key part: modeling $\STRICTINEQ$ by point configurations}

Here is the most demanding part in our chain of reductions.

\begin{theorem}\label{t:mnevv}
 $\STRICTINEQ$  reduces to {\rm CONSTRUCTIVE ORDER-TYPE REALIZABILITY}.
\end{theorem}

This result, in a somewhat different context,
was first achieved in a breakthrough by Mn\"ev (pronounce, approximately,
``Mnyoff'') \cite{Mnev-in-Rochlin}. A simplified argument was then given
by Shor \cite{Shor-stretch-hard}, and our presentation below is mostly
based on Richter-Gebert's clean treatment \cite{RichtGeb-on-Mnev}.
The same proof  method also works for reducing $\INEQ$ to 
ORDER-TYPE REALIZABILITY.

In $\STRICTINEQ$, we are given a conjunction of strict polynomial
inequalities.
The first idea is simple: simulate the evaluation
of the polynomials  by geometric constructions. 
If we represent real numbers by points on the $x$-axis, sums
and products can be constructed easily:
\immfig{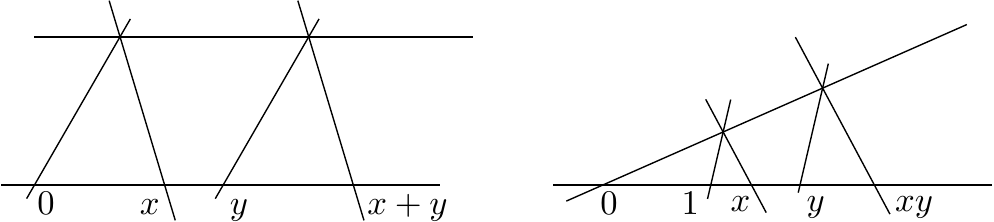}
(The parallel-looking lines should be really parallel.)

The first caveat is that these constructions use parallel lines,
which cannot be forced by order type. Moreover, given a realization
of a given order type, we can always apply a projective transformation
and get another realization, so we should better use constructions invariant
under projective transformations, and our way of representing real
numbers should be invariant as well. 

\heading{Real numbers as cross-ratios. }
Both of these issues can be remedied. First, for representing quantities
we use cross-ratios, rather than lengths. We recall that if $a,b,c,d$
are points on a line, then their \defi{cross-ratio} (in this order)
is the quantity
\[
(a,b;c,d):=\frac{|a,c|\cdot |b,d|}{|a,d|\cdot |b,c|},
\]
where $|a,b|$ denotes the \emph{oriented} Euclidean distance from $a$ to $b$,
which is positive if $a$ precedes $b$ on the line and negative otherwise.
This assumes that an orientation of the line has been chosen, but the
cross-ratio does not depend on it. More significantly, the cross-ratio
 is invariant under projective transformations. 

We thus fix a line $\ell$ on which all quantities will be represented,
and choose three points labelled $\pzero$, $\pone$, and $\pinfty$ on $\ell$; 
we say that we have chosen a \defi{projective scale} on $\ell$
(we use boldface symbols for the points of a projective scale
in order to distinguish them from the usual meaning of $0$, $1$, $\infty$).
Then a fourth point $a$ on $\ell$ represents the real number
$(a,\pone;\pzero,\pinfty)$. 
\immfig{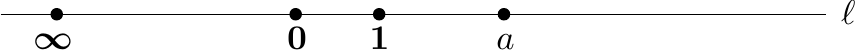}
If $\ell$ is the $x$-axis, $\pzero=(0,0)$, $\pone=(1,0)$,
 $a=(x,0)$,
and $\pinfty$ is at infinity (that is, we take the appropriate
limit in the cross-ratio), then $(a,\pone;\pzero,\pinfty)=x$, which explains
the notation.

The sum and product constructions above can be projectivized as well:
what used to be parallel lines become lines intersecting on
a distinguished line $\ell'$ (which substitutes the line at infinity).
We obtain the following \defi{von Staudt constructions}:
\immfig{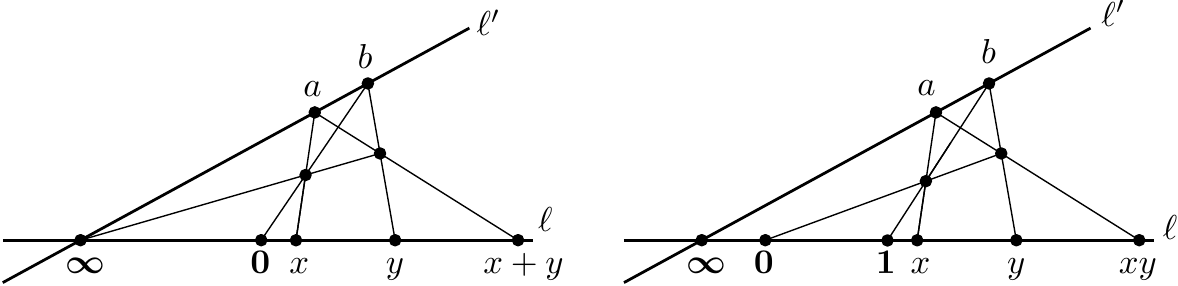}
The points $x$, $y$, $x+y$ and $xy$ are labeled by the real values
they represent. The points $\pinfty$, $\pzero$, $\pone$, $x$, and $y$
are given, $a$ and $b$ are chosen arbitrarily on the line $\ell'$,
and the remaining lines and points are constructed as indicated.
The correctness of these constructions follows from the correctness
of the constructions above with parallel lines and from the invariance
of the cross-ratio under projective transformations.

\heading{Simulating a polynomial inequality. }
Let us consider a system of strict polynomial inequalities, which we would
like to model by an order type realizability problem. 

First, in order to
have control over the ordering of the points on the line $\ell$,
we want to work only with numbers greater than $1$, and so we substitute
each variable $X_j$ with the expression $X'_j-X''_j$,
where $X'_j$ and $X''_j$ are new variables. Then, whenever 
the resulting system has a solution, it also has one
in which all variables exceed~1. From now on, we assume
that our system has this property, but we return to calling the variables
$X_1,\ldots,X_n$. 

We also transform each strict inequality $p_i>0$ in our system
to  $p_i^+>p_i^-$, where both of $p_i^+$ and $p_i^-$ are polynomials
with nonnegative integer coefficients. 
Here it is useful to have the polynomials 
in the standard form, for otherwise, it would be difficult to do this
splitting. 

After these preparations, we can calculate only with numbers greater
than~1. Let us consider one of the inequalities, 
for example $2X+1>XY^2$. We place a point $x$ representing the value
of $X$ on $\ell$, and using a series of von Staudt constructions,
we successively construct points $v_2$, $v_{2X}$, and $v_{2X+1}$,
each representing the value of the corresponding subexpression.
In a similar way, we place $y$ representing the value of $Y$
and construct $v_{Y^2}$ and $v_{XY^2}$. Then we impose
the inequality $2X+1>XY^2$ by requiring $v_{2X+1}$ and $v_{XY^2}$
to have the appropriate order along $\ell$. 
(Here it is useful to
note that all the considered number-representing points on $\ell$ may be
required to lie on the same side of $\pone$, since we can always
move $\pinfty$ sufficiently far from $\pzero$.)
This ordering can be enforced by the order type of the point configuration,
by fixing one point not lying on $\ell$ and
considering the signs of triples consisting of this point plus
two points on~$\ell$.

This looks all very nice, until one realizes
that we are still far from solving our problem. What we can already do 
is this: Given an instance of $\STRICTINEQ$, we can produce in polynomial
time a collection of points, some of them lying on $\ell$
and representing real numbers, others auxiliary, coming from the
von Staudt constructions. Certain triples of these points are
required to be collinear, which  makes the von Staudt
constructions possible,
and the sign is prescribed for some triples, which enforces the inequalities.
Realizability of such a point configuration is equivalent
to satisfiability of the given instance of $\STRICTINEQ$.

But the catch is that we are far from knowing the order type of the
configuration. For specifying it completely, we would need to know,
for example, what is the order of $v_{2X}$ and $v_{Y^2}$ on $\ell$,
and all other relations of this kind. Moreover, unless we are careful,
the various von Staudt constructions can be intermixed with one another
in an uncontrollable fashion.
This problem is  serious, and overcoming it was Mn\"ev's main achievement.

\heading{Partial order-type realizability. }
In order not to lose optimism, we point out that we have already achieved
something: we have reduced $\STRICTINEQ$ to a
decision problem that we may call PARTIAL ORDER-TYPE REALIZABILITY,
in which the combinatorial order type is given only partially, by
specifying the signs for only some of the triples. 

Actually, our partial
order type is constructible in a suitable sense, since the von Staudt
constructions produce constructible configurations, and since the
inequalities we are modeling are strict. (Indeed, constructibility would fail
if we required equalities, such as $2X+1=XY^2$, since then 
the point $v_{2X+1}=v_{XY^2}$ would have to lie on three lines
spanned by previously constructed points.) 

This reduction is already sufficient to see that 
CONSTRUCTIBLE ORDER-TYPE REALIZABILITY cannot belong to the class NP unless 
$\STRICTINEQ$, and hence all problems in $\exists\R$, do.
Indeed, assuming that CONSTRUCTIBLE ORDER-TYPE REALIZABILITY is in NP
and given an instance of $\STRICTINEQ$, we set up an instance
of CONSTRUCTIBLE PARTIAL ORDER-TYPE REALIZABILITY as above,
and then we nondeterministically guess the missing signs so that
we get a fully specified constructible combinatorial order type. 
Then we apply the the nondeterministic polynomial-time 
decision algorithm for CONSTRUCTIBLE ORDER-TYPE REALIZABILITY
whose existence we assume, and this 
yield a nondeterministic polynomial-time algorithm for $\STRICTINEQ$.
By the reductions done earlier, we also get that $\RECOG(\SEG)$
and the other problems considered along the way are not in NP
unless $\exists\R\subseteq\mathrm{NP}$. However, we do not
get $\exists\R$-completeness in this way.

\begin{proofhd}{Segment representations with large coordinates:
proof of the weaker version of Theorem~\ref{t:manydig} } By the tools
developed so far, we can also obtain segment graphs requiring coordinates
with $2^{\Omega(\sqrt n\,)}$ digits.

Instead of modeling inequalities by a point configuration as above,
we simulate  repeated squaring by the von Staudt constructions, 
obtaining points $v_1,v_2,\ldots,v_k$ on $\ell$, 
with $v_i$ representing the number
$2^{2^i}$. Since there are no unknowns, the construction can actually
be executed with some concrete points. We obtain a specific constructible
configuration of $O(k)$ points in which some four collinear points have
cross-ratio $2^{2^k}$. Moreover, \emph{every} realization
of the order type of this configuration has such a fourtuple.

Then we go through the reductions made earlier, obtaining
first a simple order type, then a description of a line arrangement,
and finally a segment graph with $n=O(k^2)$ vertices. 
From a segment representation of this graph
we can get a realization of the order type we started with---we
just follow the proofs of correctness of the reductions.
If the segment representation has integer endpoint coordinates
with at most $M$ digits, then the lines in the corresponding arrangement
have equations with $O(M)$-digit coefficients,
 and the dual simple point configuration has coordinates with $O(M)$ digits,
which we may again assume to be integers. 

The passage from the simple point configuration 
to the constructible one, as in 
the proof of Lemma~\ref{l:otrr}, is more subtle, since here the
points of the constructible configuration cannot be assumed to have
integer coordinates. Instead, we look at the ratio of the largest distance to
the smallest distance determined by the points, and we observe that
the largest distance may only decrease and the smallest only increase
in the process. This is because we always delete some points and add
a point in their convex hull. Consequently, the distance ratio in
the resulting point configuration also has $O(M)$ digits, and so
do all cross-ratios. Therefore, $M\ge 2^{\Omega(k)}=2^{\Omega(\sqrt n\,)}$.
\end{proofhd}

\begin{proofhd}
{Separating the variables and constructions: proof of
Theorem~\ref{t:mnevv}} After this detour,
we return to our main task, reducing $\STRICTINEQ$ to CONSTRUCTIBLE ORDER-TYPE
REALIZABILITY. We still need to modify the construction presented
earlier so that the order type of the resulting configuration is
determined in full, \emph{without knowing the values of the variables}.

It turns out that the crucial task is fixing the order of the points
on $\ell$; separating the von Staudt constructions is then easier.
The idea is to use not one, but 
many projective scales on $\ell$; they all have the $\pinfty$ point
in common, but the $\pzero$'s and $\pone$'s are different:
\immfig{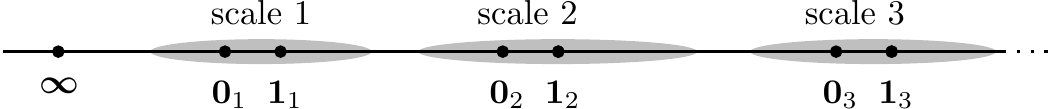}
We assume, as we may,
 that the interval available to each scale is sufficiently long
so that all numbers that need to be represented in that scale fit there.

Each of the scales $(\pinfty,\pzero_i,\pone_i)$ has distinct purpose,
and we represent  only a small number of quantities (constants or
variables) in it, for which we know the ordering. For example, each 
operation of addition will have its own scale, and so will each
multiplication, as well as each comparison of two values. 
We will also introduce gadgets, similar to the von Staudt constructions,
that link the various scales, i.e., make sure that a variable represented
in two different scales has the same value in both.

Here by variables we mean both the original variables $X_1,\ldots,X_n$
from the formula and variables $V_\xi$, where $\xi$ runs through all
subexpressions to be evaluated. 
Moreover, and this is the final trick,
we also introduce variables $V_{-\xi}$ and $V_{1/\xi}$ for every
subexpression $\xi$, as well as $V_{-X_i}$ and $V_{1/X_i}$ for every 
variable~$X_i$. 

Let us see, for example, how we should process the inequality $XY+Z>Y$.
\begin{enumerate} 
\item We place $\pzero_1$ and $\pone_1$ defining scale~1.
Then we place a point $x$ representing $X$
to the right of $\pone_1$ (since, as we recall, all of the variables $X_i$
and subexpressions $\xi$ have values in $(1,\infty)$). Still in scale~1,
we construct a point $1/x$ representing $V_{1/X}=1/X$, 
by a suitable inversion gadget---see Fig.~\ref{f:j-pgadgets}.
\begin{figure}[tbp]
  \centering%
  {\includegraphics{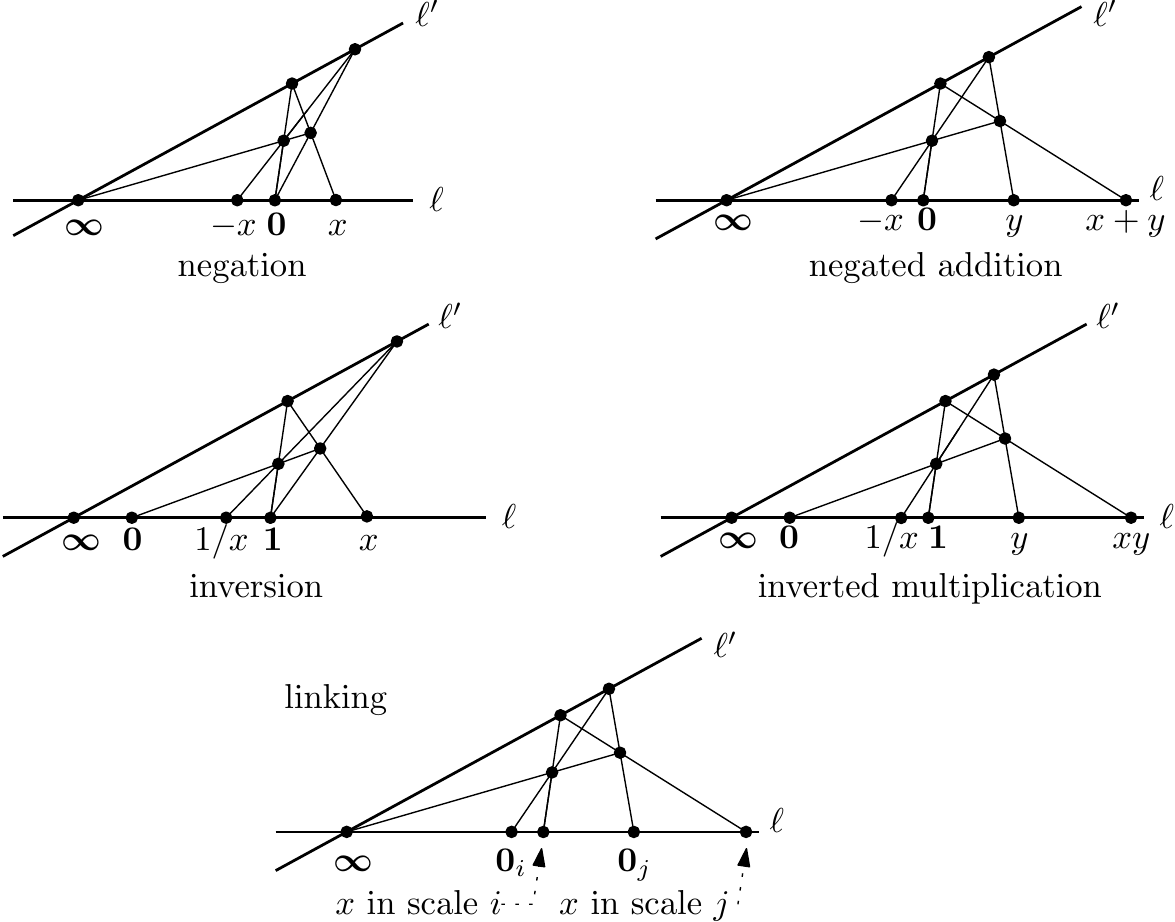}}
  \caption{The construction gadgets. \label{f:j-pgadgets}}
\end{figure}
\item We place $\pzero_2$ for scale~2, and obtain $\pone_2$
by transferring $\pone_1$ from scale~1 using a linking gadget.
Then in scale~$2$, we place $y$ representing $Y$, we transfer $V_{1/X}$ there
from scale~1 by a linking gadget, and by an inverted multiplication gadget
we produce a point representing $V_{XY}=XY$. The point is that we do not
care which of $X$ and $Y$ is larger, since they never appear
in the same scale; in scale~2, we have only $1/X$, $Y$, and $XY$, for which
the ordering is obvious in view of $X,Y>1$.
\item We initialize scale~3 by placing $\pzero_3$ and obtaining
$\pone_3$ by a linking gadget from scale~1 (or 2, both work).
We place a point representing $Z$ and construct
$-z$ representing $V_{-Z}=-Z$. 
\item We similarly initialize scale~4, and
by two linking gadgets, we transfer $V_{XY}$ and $V_{-Z}$ there.
By a negated addition gadget we construct
a point representing $V_{XY+Z}$ there.
\item We initialize scale~5,
we transfer $V_{XY+Z}$ and $Y$ to it, and we enforce $V_{XY+Z}>Y$ there.
\end{enumerate}

In this way, the ordering of the points on $\ell$ can be fixed,
and it remains to untangle the various gadgets. Each gadget
has four auxiliary points; let us call them $a_i,b_i,c_i,d_i$ for
the $i$th gadget used in the construction.

The points $a_i$ and $b_i$ lie on the line $\ell'$
through $\pinfty$, which is shared by all of the gadgets.
Their position on $\ell'$ can be chosen freely, and in order to
separate the gadgets, we place $a_i$ at distance $D_i$ from
$\pinfty$ and $b_i$ at distance $\eps_i$ from $a_i$,
where the sequence $D_1\ll D_2\ll \cdots$ increases extremely fast,
and $\eps_1\gg \eps_2\gg\cdots$ decreases even faster.

The points $c_i$ and $d_i$ are then determined by the ``input values''
of the gadget, and if $a_i$ and $b_i$ are sufficiently close,
$c_i$ and $d_i$ lie in an arbitrarily small neighborhood
of $a_i$ and $b_i$.
\immfig{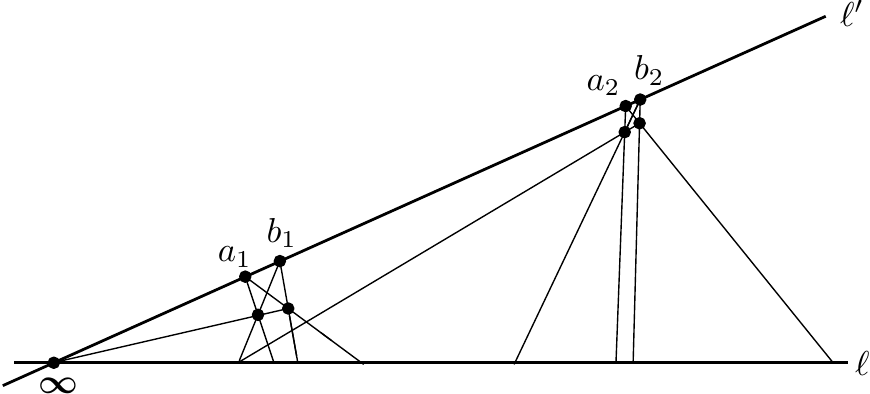}

Since we know the order of the points on $\ell$, as well as the points
 where the six lines spanned by each fourtuple $a_i,\ldots,d_i$
intersect $\ell$, we know the sign of every triple involving
points on $\ell$ and auxiliary points from a single gadget.
By similar considerations, we can determine the sign of all of the
remaining types of triples (on $\ell$ $+$ gadget $+$ another gadget,
two same gadget $+$ another gadget, three different gadgets). 

In this way, we arrive at a fully specified and constructible 
combinatorial order type, whose realizability is equivalent to
the solvability of the given instance of $\STRICTINEQ$.
This concludes the proof of Theorem~\ref{t:mnevv}.
\end{proofhd}

\heading{The Mn\"ev universality theorem. } Mn\"ev's original result
did not deal with computational complexity---it was actually about a 
topological question. To formulate it, let us consider a
sequence $(p_1,p_2,\ldots,p_n)$ of $n$ points in the plane.
This sequence is specified by a list of the point coordinates,
i.e., $2n$ real numbers, and so we can regard it as a single point
in $\R^{2n}$. For a combinatorial order type $\tau$ of $n$ points,
we define the \defi{realization space} $\RR(\tau)\subseteq \R^{2n}$ as the set
of the points corresponding to realizations of~$\tau$.

In 1956, Ringel asked, in the equivalent language of line arrangements,
 whether the realization space has to be path-connected; in other words,
whether one realization of $\tau$ can always be continuously deformed into
any other, while keeping the order type $\tau$ along the way. (Strictly
speaking, we need to be somewhat careful about mirror reflection, which
may trivially disconnect the realization space as defined above into
two components---but one usually factors out affine transformations,
by fixing the position of three affinely independent points.)

Mn\"ev's universality theorem shows that the answer to Ringel's question
is no in the strongest possible sense: the realization space can be 
topologically as complicated as one may wish. It may have the ``shape''
of any prescribed semialgebraic set, or of any prescribed finite simplicial
complex; the appropriate term from topology for ``having the same shape''
here is being \emph{homotopy equivalent}. Mn\"ev's statement uses an even
stronger notion of \emph{stable equivalence}, for which we refer to
\cite{RichtGeb-on-Mnev}. 

Similar universality theorems also hold
for the realization spaces of other kinds of objects from $\exists\R$-complete
problems. Essentially, $\exists\R$-completeness and topological universality
theorems are just two ways of expressing a great intrinsic complexity
of a given class of objects. The proofs are also similar, although
in topology one has to watch out for other aspects than in computational
complexity.

\section{Quantifier elimination according to Muchnik}\label{s:muchnik}

Here we prove that the first-order theory of $\R$, and thus, in particular,
the existential theory of $\R$, are \emph{decidable}.
We will actually obtain a stronger result, referred to as 
\defi{quantifier elimination}. 

\begin{theorem}[Tarski  \cite{t-dmeag-51}]\label{t:tarski}
There is an algorithm accepting as an input a 
formula $\Psi$ of the first-order theory of the reals, 
which may contain quantifiers; in general it also
contains free variables $Y_1,\ldots,Y_n$, which we
write as $\Psi=\Psi(\bY)$, with $\bY=(Y_1,\ldots,Y_n)$. 
The algorithm outputs a \emph{quantifier-free}
formula $\Phi=\Phi(\bY)$ that is equivalent to $\Psi$;
that is, for every choice of $\yy\in\R^n$
we have $\Psi(\yy)\equiv\Phi(\yy)$.
\end{theorem}

Geometrically, this result tells us that every subset of $\R^n$ that
can be described in the first-order theory of $\R$ is semialgebraic,
i.e., can be specified by a quantifier-free formula. Even more
geometrically, the essence of the quantifier-elimination result
is that the projection of a semialgebraic set onto a coordinate
subspace is again semialgebraic.

The running time bounds for the algorithm presented below are actually poor, 
much worse than for the best known algorithms, and we will not care 
about them. However, the algorithm is relatively simple and it exhibits
some of the ideas also appearing in more sophisticated algorithms.

The main ideas of this algorithm are due to Muchnik (unpublished);
our presentation is mostly based on \cite{MichauxOzturk}, also drawing
inspiration from the blog~\cite{bhatta-blog}.

We assume that the given formula $\Psi$ is in prenex form, with all quantifiers
on the outside. We eliminate the quantifiers one by one starting from inside.
A universal quantifier is converted into an existential one
using $(\forall X)\Phi\equiv \neg(\exists X)\neg\Phi$.
Thus, it suffices to describe a procedure for eliminating the single
existential quantifier from a formula $\Psi$ of the form
\[
(\exists X)F(A_1,\ldots,A_m),
\]
where $F$ is a Boolean formula, each $A_i$ is an atomic predicate
of the form $p_i(X,\bY)~\mathrm{rel}~0$  for some 
polynomial $p_i$ with integer coefficients, and $\bY=(Y_1,\ldots,Y_n)$
is the vector of the free variables of~$\Psi$.

\subsection{The univariate case } 

In order to explain the method,
it is instructive to start with the case $n=0$, where $\Psi$
is a sentence
(no free variables) and we deal with univariate polynomials
$p_1(X),\ldots,p_m(X)$. 

For $x\in\R$, let the \defi{sign vector} of $(p_1,\ldots,p_m)$ at $x$
be
\[
\bigl(\sgn p_1(x),\sgn p_2(x),\ldots,\sgn p_m(x)\bigr)\in\{-1,0,1\}^m.
\]
To decide the validity of the sentence $\Psi$, we want to know
whether there is an $x\in\R$ at which the sign vector attains
one of the values allowed by the formula. For example,
for the formula
\[
(\exists X)\bigl(p_1(X)>0\wedge (p_2(X)>0\vee p_3(X)=0)\bigr),
\]
the allowed sign vectors are $(1,1,*)$ and $(1,*,0)$, where
$*$ means ``arbitrary value''.

The algorithm actually 
computes \emph{all} possible sign vectors of $(p_1,\ldots,p_m)$;
this is a common feature of practically all known algorithms for 
quantifier elimination or for deciding $\ETR$. Thus, the only way the
Boolean formula $F$ enters the computation is in checking if any
of the resulting sign vectors is allowed by the formula.

\heading{The sign table. }
Let us consider an example with three polynomials
$p_1(X)=4-X^2$, $p_2(X)=X^3-2X^2-X+2$, and $p_3(X)=-X^3+5X^2-6X$. Their
graphs are plotted in the next picture, 
\immfig{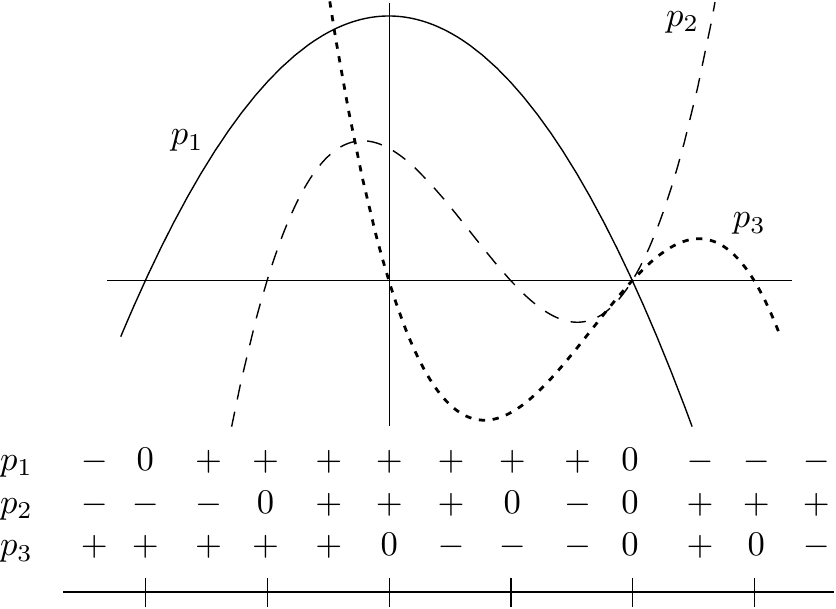}
and below the picture we have the \defi{sign table} of the ordered
triple $(p_1,p_2,p_3)$. To make the sign table, we partition the real line
into maximal subintervals so that each $p_i$ has a constant sign on each
of the subintervals. Thus, open subintervals alternate with one-point subintervals;
we call the latter ones the \defi{boundaries}.
Each subinterval $I$ is then labeled by the sign vector of the $p_i$
at an (arbitrary) point $x\in I$, and the sign table is the sequence
of these sign vectors ordered from left to right.
We stress that the sign table is \emph{only} the sequence of sign vectors,
and it does not include the numerical values of the boundaries
(which are roots of the~$p_i$). 

We want to compute the sign table, which is of course sufficient
to decide the given sentence $\Psi$, and for reasons
which will become apparent later, we want to do so
\emph{without} computing the roots of our polynomials 
(although the sign table provides the number of roots
of each $p_i$ and their relative positions).

The algorithm proceeds incrementally.
Having computed the sign table for $(p_1,p_2,\ldots,p_{k-1})$, we want to
extend it to the sign table for $(p_1,p_2,\ldots,p_k)$. Conceptually,
we divide this into two steps:
\begin{enumerate}
\item[(i)] We compute the sign of $p_k$ at the roots of $p_1,\ldots,p_{k-1}$,
i.e., at the boundaries of the old sign table.
\item[(ii)] We locate the roots of $p_k$ among the old boundaries, 
and we add new boundaries and sign vectors accordingly.
\end{enumerate}
Both of the steps should be far from obvious for now. But 
let us assume for a moment that Step~(i) has been accomplished somehow.
Thus, for example, we have the sign table for $(p_1,p_2,p_3)$
as in the example above, and we also have the signs of another
polynomial $p_4$ at the boundaries:
\immfig{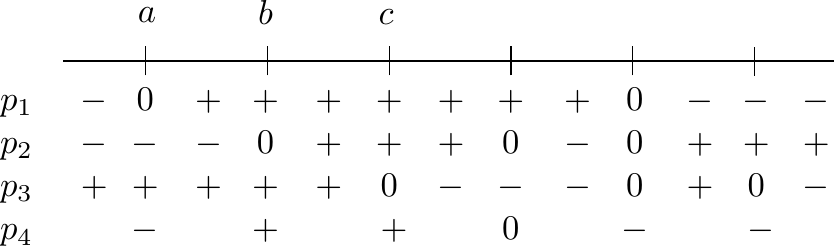}  
What should be the columns of the new sign table, after
adding $p_4$, between the two boundary columns $a$ and $b$ of the old table? 
Since $p_4(a)<0$ and $p_4(b)>0$, there must be a root of $p_4$ in
the interval $(a,b)$, so a reasonable guess at the new columns is
\immfig{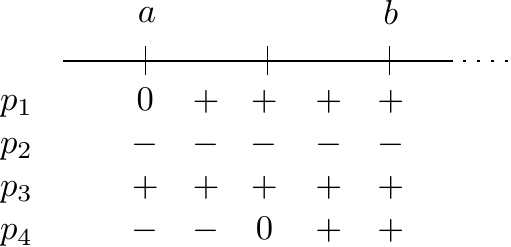}
This need not be correct, though, since $p_4$ may have \emph{several}
roots in $(a,b)$. But if it does, then it must have a local minimum there,
and hence the derivative $p'_4$ has a root in $(a,b)$. 

So if we could make sure that all roots of $p'_4$ are among the boundaries
of the old sign table, our guess at the new columns as above would 
be correct. Under the same assumption, for the next subinterval $(b,c)$, 
from $p_4(b)>0$ and $p_4(c)>0$ 
we would get that $p_4>0$ on $(b,c)$. Similarly we could infer all of the new
sign table. (There is a detail that we have skipped---we also need
the signs of $p_4$ ``at $-\infty$'' and ``at $+\infty$,'' i.e.,
for all sufficiently small and sufficiently large $x$, for otherwise,
we would not be able to detect roots of $p_4$ to the left 
or to the right of all old boundaries.)

This reasoning works in general, for adding a new polynomial $p_k$,
but how do we force all roots of $p'_k$ to be among the old
boundaries? The answer is that, at the beginning of
the algorithm, we \emph{extend} our collection of polynomials by adding 
their derivatives, second derivatives, etc. 
Then it suffices to process the polynomials
from smaller degrees to larger ones, and the above condition 
on $p'_k$ will always be met! 

Very good, so how do we perform Step~(i), determining the
sign of $p_k$ at the old boundaries?  Suppose that we
want to determine $\sgn p_k(\alpha)$, where $\alpha$ is a root
of $p_i$, $i<k$. The trick is to consider the \defi{remainder}
$r:=p_k\,\mathrm{mod}\, p_i$; i.e.,
$r$ is the unique polynomial with $\deg r<\deg p_i$ and
$p_k=qp_i+r$ for some polynomial~$q$. (We note that if $p_k$ and $p_i$ have
integer coefficients, $r$ may still have rational but non-integral
coefficients.)

We have
\[ r(\alpha)=p_k(\alpha)-q(\alpha)p_i(\alpha)= p_k(\alpha)\]
since $p_i(\alpha)=0$, and so, in particular,
$\sgn p_k(\alpha)=\sgn r(\alpha)$.
Thus, if $r$ is among the already processed polynomials $p_1,\ldots,p_{k-1}$,
we know its sign at the old boundaries, and so we can infer the sign of $p_k$
(without computation). This suggests that, similar to the
case with the derivatives, at the beginning of the algorithm
we also extend our collection of polynomials  by closing it
under taking remainders. 

\heading{Computing the closure. } Let $\PP=\{p_1,\ldots,p_m\}$
be the initial collection of polynomials, those in the given formula.
Let us define $\overline{\PP}$ as the closure of $\PP$
under taking the derivative and remainders, i.e., an inclusion-minimal
set of polynomials containing $\PP$ and closed under these two operations.

\begin{exercise}\label{ex:closure'rem}
 {\rm (a)} Describe an algorithm to compute $\overline{\PP}$ given $\PP$, and prove its finiteness (assuming that subroutines for computing
the derivative and remainder are available).

{\rm (b)}  What lower and upper bounds can you get for $|\overline{\PP}|$
in terms of $m$ and $\Delta$, the maximum degree of the~$p_i$?
\end{exercise}

\heading{Summary of the univariate 
sign-table algorithm. } The input to the algorithm
is a set $\PP$ of univariate polynomials with integer
coefficients, and the output is a sign table for a certain superset of~$\PP$.

\begin{enumerate}
\item We compute the closure $\overline\PP$ of $\PP$ under derivatives
and remainders as in Exercise~\ref{ex:closure'rem}.
We number the polynomials in $\overline\PP$ as $p_1,p_2,\ldots,p_M$
in such a way that $\deg p_i\le\deg p_j$ for $i<j$. (This numbering
is different from the original numbering of the input polynomials
in $\PP$; what used to be $p_3$ may now be $p_{1234}$.)
\item Let the sign table $T_0$ consist of a single empty sign vector.
For $k=1,2,\ldots,M$, compute $T_k$, the sign table for $(p_1,\ldots,p_k)$,
from $T_{k-1}$ according to the next step, and then output $T_M$.
\item If $\deg p_k=0$, i.e., $p_k$ is a constant, then determine its sign
and add the appropriate row to $T_{k-1}$. Otherwise, for 
$\deg p_k\ge 1$, do the two steps as above:
\begin{enumerate}
\item[(i)] Determine the sign of $p_k$ at $-\infty$ and at $+\infty$,
by inspecting the leading coefficient and the degree. 
Then for every boundary $a$ of $T_{k-1}$, find $p_i$, $i<k$, with $p_i(a)=0$,
and determine $\sgn p_k(a):=\sgn r(a)$, where $r:=p_k\,\mathrm{mod}\,p_i$.
Here $\sgn r(a)$ can be found in the sign table $T_{k-1}$, since
$r=p_j$ for some~$j<k$.
\item[(ii)] \ifafour\else\begin{sloppypar}\fi
For every interval $(a,b)$, where $a,b$ are consecutive boundaries
of $T_{k-1}$ (or $a=-\infty$ or $b=\infty$), inspect the signs of $p_k(a)$
and $p_k(b)$, and determine the behavior of $p_k$ on $(a,b)$:
If $(\sgn p_k(a))(\sgn p_k(b))=-1$, then $p_k$ has a single root in $(a,b)$
and we insert a new column into the sign table as described in the example
above. In all other cases, $\sgn p_k$ is constant and nonzero on
$(a,b)$, and it coincides with the nonzero
sign in  $\{\sgn p_k(a),\sgn p_k(b)\}$ (at least one of these must
be nonzero---why?).
 \ifafour\else\end{sloppypar}\fi
\end{enumerate}
\end{enumerate}

\subsection{The multivariate case. }

Now we consider the general case, where we want to eliminate the
existential quantifier from $(\exists X)\Phi_0(X,\bY)$, 
where $\bY=(Y_1,\ldots,Y_n)$ and $\Phi_0$ is quantifier-free.
Let us think of the polynomials $p_i(X,\bY)$ appearing in 
$\Phi_0$ as polynomials in $X$ whose coefficients are polynomials
in $Y_1,\ldots,Y_n$.

The idea, a very computer-science one, is to run the univariate
sign-table algorithm on the polynomials $p_1(X,\bY),\ldots,p_m(X,\bY)$.
Of course, that algorithm was derived for polynomials with integer
(or rational) coefficients, while here the coefficients are
polynomials. We thus need to look at the operations performed
in the algorithm in more detail.

One thing the univariate algorithm does is computing various
expressions in the coefficients using the four basic operations
$+,-,*,/$ (no square roots etc.). That we can do with polynomials as
well, obtaining rational functions in the $Y_j$, of the form
$u(\bY)/v(\bY)$, where $u(\bY)$ and $v(\bY)$ are 
polynomials.\footnote{Actually, the only place where division appears
is the computation of remainders. As in \cite{MichauxOzturk}
and elsewhere in algorithms dealing with multivariate polynomials,
one can avoid division completely by replacing
the remainder by the \defi{pseudoremainder}: the pseudoremainder
of polynomials $a(X)=\sum_{i=0}^d a_iX^i$ and 
$b(X)=\sum_{i=0}^e b_iX^i$, $d\ge e$, is the unique polynomial $r(X)$ such that
$b_e^{d-e+1} a(x)= q(X)b(X)+r(X)$.}

The other kind of operation in the algorithm is \emph{sign testing}:
for some already computed expression $E$, a rational function in our case,
the algorithm wants to know the sign of~$E$. 
When we try to work with
rational functions instead of rational numbers, we do not know the answer,
since it generally depends on the values of 
the~$Y_j$.\footnote{We should remark that sign testing may also occur 
in perhaps somewhat
unexpected parts of the algorithm. For example, when computing the
closure $\overline\PP$, the algorithm repeatedly computes the remainder
for two polynomials computed earlier, and for that, it needs to
know the degrees of these polynomials. If the coefficients of
the polynomials are rational functions of $\bY$, then the leading
coefficient, as well as some others,
 may vanish for some values of $\bY$, and so the algorithm must test which
is the highest power of $X$ whose coefficient is actually nonzero.}

Here comes the key idea of the multivariate case: we just give 
\emph{all three} possible answers to the algorithm. Then, instead of
following just one computation, we go into three different
branches of computation (this may delight fans of
the many-worlds interpretation of quantum mechanics). 

Such a branching is made at every sign test, so we get a ternary tree,
in which every possible computation of the algorithm corresponds to
one of the root-leaf paths, and each node is labeled by a rational
function $E=E(\bY)$ from the corresponding sign test. Here is an
(artificial and very small) example of such a tree:
\immfig{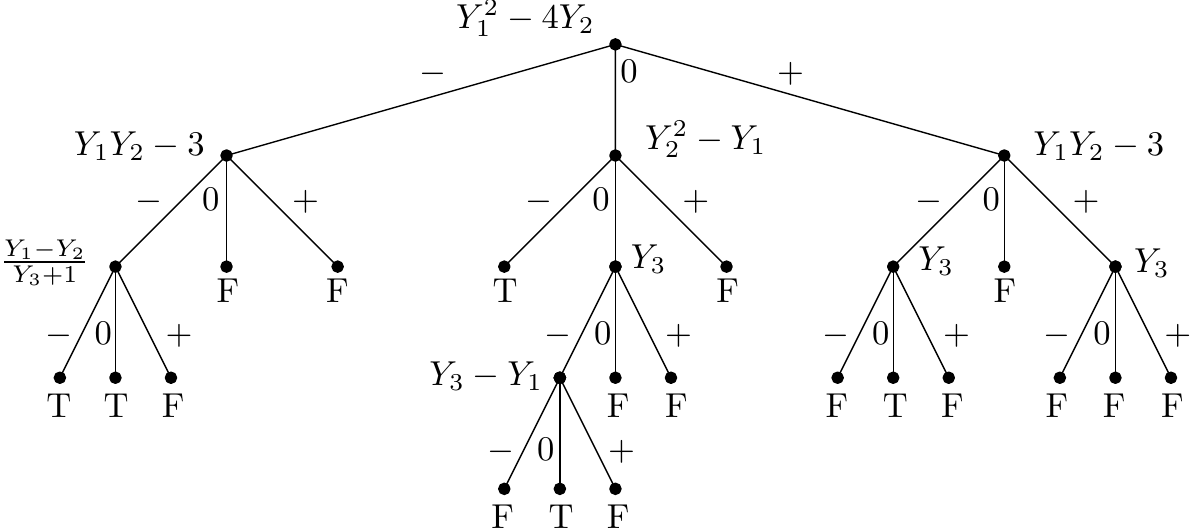}
For a real but still simple example see Exercise~\ref{ex:quadr} below.

At the end of each computation, the univariate algorithm outputs
TRUE or FALSE, depending on whether the computed sign table contains
one of the sign vectors allowed by the formula $\Phi_0$. We label the
leaves of the tree by these output values.

Now we recall our actual goal: we want a quantifier-free formula $\Phi(\bY)$
equivalent to $(\exists X)\Phi_0(X,\bY)$. We construct $\Phi(\bY)$
as a formula describing the conditions under which the univariate
algorithm reaches a leaf labeled TRUE.  This is easy to do: 
We make $\Phi(\bY)$ as a disjunction of subformulas, each corresponding
to one TRUE leaf. The subformula corresponding to a given leaf $\ell$
should say that all the rational functions along the path to $\ell$
have the signs given by the chosen path. For example, the
path to the leftmost leaf in the picture above yields the
subformula
\[
Y_1^2-4Y_2<0\wedge Y_1Y_2-3<0 \wedge \frac{Y_1-Y_2}{Y_3+1}<0.
\]

This is not yet a formula in the first-order theory of the reals,
since it contains division. But we can replace sign testing for
rational functions by sign testing for polynomials, by testing
the numerator and the denominator separately.

This finishes the presentation of the quantifier-elimination
algorithm and the proof of Theorem~\ref{t:tarski}.

\begin{exercise}\label{ex:quadr} Apply the algorithm to the formula
$\Psi(A,B,C)\equiv (\exists X)AX^2+BX+C=0$, and see how the
familiar discriminant materializes in front of your eyes.
\end{exercise}

\subsection{On the complexity of quantifier-elimination algorithms}

Let the considered formula $\Psi$ contain $m$ polynomials of 
degree at most $\Delta$ each and with the number of bits
in each coefficient bounded by $\tau$. Further let it have
$\omega$ blocks of alternating quantifiers, with $n_i$ variables
in the $i$th block, and let $n_0$ be the number of free variables.
Let us set $N:=\prod_{i=0}^\omega (n_i+1)$, and let $C$ be a suitable constant.
According to \cite{BPR-QE,BasuPollackRoy-book}, there is an algorithm
that performs quantifier elimination for $\Psi$
using $m^{N} (\Delta+1)^{C^\omega N}$ arithmetic operations with
at most $((\Delta+1)^{C^\omega N}\tau)$-bit integers. 

For deciding an existential formula (i.e., for $n_0=0$ and $\omega=1$),
this yields an algorithm with $m^{n+1}(\Delta+1)^{O(n)}$ arithmetic
operations with integers having at most $(\Delta+1)^{O(n)}\tau$ bits.

\heading{Why quantifier elimination is doubly exponential.}
The following exercise provides an example of a quantified formula
of length $O(n)$ for which any quantifier elimination has to produce 
a quantifier-free formula of length doubly exponential in~$n$.
It is a simplified variant of an example by Davenport and Heintz
\cite{DavenportHeintz}.

\begin{exercise} Let the formula $\Psi_n(X,Y)$ be defined inductively:
$\Psi_0(X,Y)$ is  $Y=4X(1-X)$, and $\Psi_n(X,Y)$ is
\[
(\exists Z)(\forall U\, V)
\bigl((U=X\wedge V=Z)\vee (U=Z\wedge V=Y)\bigr)\Rightarrow
\Psi_{n-1}(U,V)
\]
(we have used the implication $\Rightarrow$, which strictly speaking
wasn't mentioned among the Boolean connectives appearing in
the considered formulas, but $A\Rightarrow B$ can be
replaced by $\neg A\vee B$). 

{\rm (a)} The formula $\Psi_0(X,Y)$ defines the \emph{logistic map}
$f(x)=4x(1-x)$. What map is defined by $\Psi_n(X,Y)$?

{\rm (b)} Show that the formula $\tilde \Psi_n(X)$, defined
as $\Psi_n(X,\frac12)$ (wait: fractions were not allowed
in formulas of the first-order theory of the reals,
so how do we write this properly?), defines a semialgebraic
subset of $\R$ whose definition by a quantifier-free formula
requires formula length $2^{2^{\Omega(n)}}$. You may want to use
Exercise~\ref{ex:1d-semi}.
\end{exercise}

\bibliographystyle{alpha}
\bibliography{gg}

\newcommand{\etalchar}[1]{$^{#1}$}
\begin{thebibliography}{BLVS{\etalchar{+}}99}

\bibitem[Bha12]{bhatta-blog}
Arnab Bhattacharyya.
\newblock Something you should know about: Quantifier elimination ({Part}
  {II}).
\newblock Blog, {\tt
  http://cstheory.\-blogoverflow.\-com/2012/02/\-something-you-should\--know-about-quantifier\--elimination-part-ii},
  2012.

\bibitem[Bie91]{DBLP:journals/dcg/Bienstock91}
Daniel Bienstock.
\newblock Some provably hard crossing number problems.
\newblock {\em Discrete {\&} Computational Geometry}, 6:443--459, 1991.

\bibitem[BLVS{\etalchar{+}}99]{Bjorner:Oriented}
Anders Bj{\"o}rner, Michel Las~Vergnas, Bernd Sturmfels, Neil White, and
  G{\"u}nter~M. Ziegler.
\newblock {\em Oriented Matroids (Second Edition)}.
\newblock Cambridge University Press, Cambridge, 1999.

\bibitem[BPR96]{BPR-QE}
Saugata Basu, Richard Pollack, and Marie-Fran\c{c}oise Roy.
\newblock {On the combinatorial and algebraic complexity of quantifier
  elimination}.
\newblock {\em J. ACM}, 43(6):1002--1045, 1996.

\bibitem[BPR03]{BasuPollackRoy-book}
Saugata Basu, Richard Pollack, and Marie-Fran\c{c}oise Roy.
\newblock {\em Algorithms in real algebraic geometry}.
\newblock Algorithms and Computation in Mathematics 10. Springer, Berlin, 2003.

\bibitem[BV07]{BasuVorobjov-homot}
Saugata Basu and Nicolai~N. Vorobjov.
\newblock {On the number of homotopy types of fibres of a definable map}.
\newblock {\em J. Lond. Math. Soc., II. Ser.}, 76(3):757--776, 2007.

\bibitem[Can88]{Canny-pspace}
John Canny.
\newblock Some algebraic and geometric computations in {PSPACE}.
\newblock In {\em Proc. 20th Annu. ACM Sympos. Theory Comput.}, pages 460--467,
  1988.

\bibitem[CG09]{ChalopinGoncalves}
J{\'e}r{\'e}mie Chalopin and Daniel Gon\c{c}alves.
\newblock Every planar graph is the intersection graph of segments in the
  plane.
\newblock In {\em Proc. 41st Annual ACM Symposium on Theory of Computing
  (STOC)}, pages 631--638, New York, NY, USA, 2009. ACM.

\bibitem[DH88]{DavenportHeintz}
James~H. Davenport and Joos Heintz.
\newblock Real quantifier elimination is doubly exponential.
\newblock {\em J. Symbolic Comput.}, 5(1-2):29--35, 1988.

\bibitem[Gj88]{GrigorievVorobjov88}
Dima~Yu. Grigor'ev and Nikolai N.~Vorobjov jun.
\newblock {Solving systems of polynomial inequalities in subexponential time.}
\newblock {\em J. Symb. Comput.}, 5(1-2):37--64, 1988.

\bibitem[GPS90]{gps-iscir-90}
Jacob~E. Goodman, Richard Pollack, and Bernd Sturmfels.
\newblock The intrinsic spread of a configuration in {$\Re^d$}.
\newblock {\em J. Amer. Math. Soc.}, 3:639--651, 1990.

\bibitem[KM94]{km-igs-94}
Jan Kratochv{\'\i}l and Ji{\v{r}}{\'i} Matou{\v s}ek.
\newblock Intersection graphs of segments.
\newblock {\em J. Combin. Theory Ser. B}, 62(2):289--315, 1994.

\bibitem[KM12]{DBLP:journals/dcg/KangM12}
Ross~J. Kang and Tobias M{\"u}ller.
\newblock Sphere and dot product representations of graphs.
\newblock {\em Discrete {\&} Computational Geometry}, 47(3):548--568, 2012.

\bibitem[MM13]{McDiarmidMuellerDisks}
Colin McDiarmid and Tobias M{\"u}ller.
\newblock {Integer realizations of disk and segment graphs.}
\newblock {\em J. Comb. Theory, Ser. B}, 103(1):114--143, 2013.

\bibitem[Mne89]{Mnev-in-Rochlin}
Nikolai~E. Mnev.
\newblock The {universality} theorems on the classification problem of
  configuration varieties and convex polytopes varieties.
\newblock In O.~Ya. Viro, editor, {\em Topology and Geometry - Rohlin Seminar},
  volume 1346 of {\em Lecture Notes Math.}, pages 527--544. Springer-Verlag,
  1989.

\bibitem[MO02]{MichauxOzturk}
Christian Michaux and Adem Ozturk.
\newblock Quantifier elimination following {M}uchnik.
\newblock Univ. de Mons-Hainaut Preprint Series (\#10), 2002.

\bibitem[PKK{\etalchar{+}}12]{pawlik2012triangle}
Arkadiusz Pawlik, Jakub Kozik, Tomasz Krawczyk, Micha{\l} Laso{\'n}, Piotr
  Micek, William~T. Trotter, and Bartosz Walczak.
\newblock Triangle-free intersection graphs of line segments with large
  chromatic number.
\newblock Preprint, arXiv:1209.1595, 2012.

\bibitem[RG95]{RichtGeb-on-Mnev}
J{\"u}rgen Richter-Gebert.
\newblock Mn{\"e}v's universality theorem revisited.
\newblock \emph{S\'emin. Lothar. Comb.} (electronic), B34h, 1995.

\bibitem[RG11]{richter2011perspectives}
J{\"u}rgen Richter-Gebert.
\newblock {\em Perspectives on projective geometry: a guided tour through real
  and complex geometry}.
\newblock Springer Verlag, Berlin Heidelberg, 2011.

\bibitem[Sch10]{Schaefer-surv-exR}
Marcus Schaefer.
\newblock {Complexity of some geometric and topological problems}.
\newblock In {\em {Eppstein, David (ed.) et al., Graph drawing. 17th
  international symposium, GD 2009, Chicago, IL, USA, September 22--25, 2009.
  Lecture Notes in Computer Science 5849}}, pages 334--344. Springer, Berlin,
  2010.

\bibitem[Sho91]{Shor-stretch-hard}
Peter~W. Shor.
\newblock Stretchability of pseudolines is {NP}-hard.
\newblock In P.~Gritzman and B.~Sturmfels, editors, {\em Applied Geometry and
  Discrete Mathematics: The Victor Klee Festschrift}, volume~4 of {\em DIMACS
  Series in Discrete Mathematics and Theoretical Computer Science}, pages
  531--554. AMS Press, 1991.

\bibitem[S{\v{S}}11]{SchaeStef-Nash}
Marcus Schaefer and Daniel {\v{S}}tefankovi{\v{c}}.
\newblock Fixed points, {N}ash equilibria, and the existential theory of the
  reals.
\newblock Manuscript, DePaul University, Chicago, 2011.

\bibitem[Tar51]{t-dmeag-51}
Alfred Tarski.
\newblock {\em A decision method for elementary algebra and geometry}.
\newblock Univ. of California Press, Berkeley, CA, 1951.

\end{thebibliography}

\end{document}